\documentclass[reprint,amsmath,amssymb,aps,pra, nofootinbib]{revtex4-2}
\usepackage{graphicx,color}
\usepackage{dcolumn}

\usepackage{amssymb,amsmath,amsfonts,bm, bbm}
\usepackage{amsfonts,bm, bbm}
\usepackage{graphicx}
\usepackage{textcomp}
\usepackage{amsthm}
\usepackage{mathtools}
\usepackage[ruled,vlined,linesnumbered]{algorithm2e}
\usepackage{enumitem}
\usepackage{physics}
\usepackage[export]{adjustbox}
\usepackage[normalem]{ulem}
\usepackage{float}
\mathtoolsset{showonlyrefs=true}
\usepackage[nodisplayskipstretch]{setspace}

\theoremstyle{plain}
\newtheorem{definition}{Definition}[]
\newtheorem{theorem}{Theorem}[]

\newtheorem{corollary}{Corollary}[]

\newtheorem{lemma}{Lemma}[]
\newtheorem{remark}{Remark}[]

\usepackage[dvipsnames]{xcolor}

\renewcommand{\qedsymbol}{$\blacksquare$}

\newcommand{\mc}{\mathcal}
\newcommand{\lt}{\left}
\newcommand{\rt}{\right}

\usepackage{hyperref}

\begin{document}

\preprint{APS/123-QED}

\title{Reliable Quantum Memories with Unreliable Components}
\author{Anuj K.\ Nayak}
\author{Eric Chitambar}
\author{Lav R.\ Varshney}
\affiliation{
Coordinated Science Laboratory, University of Illinois at Urbana-Champaign, Urbana, IL 61801, USA
}


\begin{abstract}
Quantum memory systems are vital in quantum information processing for dependable storage and retrieval of quantum states. Inspired by classical reliability theories that synthesize reliable computing systems from unreliable components, we formalize the problem of reliable storage of quantum information using noisy components. 
We introduce the notion of stable quantum memories and define the storage rate as the ratio of the number of logical qubits to the total number of physical qubits, as well as the circuit complexity of the decoder, which includes both quantum gates and measurements. We demonstrate that a strictly positive storage rate can be achieved by constructing a quantum memory system with quantum expander codes. Moreover, by reducing the reliable storage problem to reliable quantum communication, we provide upper bounds on the achievable storage capacity. In the case of physical qubits corrupted by noise satisfying hypercontractivity conditions, we provide a tighter upper bound on storage capacity using an entropy dissipation argument. {\color{black}Furthermore, observing that the time complexity of the decoder scales non-trivially with the number of physical qubits, achieving asymptotic rates may not be possible due to the induced dependence of the noise on the number of physical qubits. In this constrained non-asymptotic setting, we derive upper bounds on storage capacity using finite blocklength communication bounds.} 
Finally, we numerically analyze the gap between upper and lower bounds in both asymptotic and non-asymptotic cases, and provide suggestions to tighten the gap. 
\end{abstract}

\maketitle

\section{Introduction} 

Quantum memory systems play a crucial role in quantum information processing, enabling the reliable storage and retrieval of quantum states \cite{qm2023, microarch2023}. Quantum information storage is a critical element for a variety of computing and communication protocols, providing a foundation for secure and efficient quantum operations.

Synthesis of reliable classical computing systems using unreliable noisy components was first studied by von Neumann \cite{von1956probabilistic}. Construction of classical reliable memory systems using noisy registers and noisy logical gates was investigated in the works of Taylor \cite{taylor1968reliable} and Kuznetsov \cite{kuznietsov1973information}. Necessary and sufficient conditions for reliable storage were improved in subsequent works \cite{varshney2015toward, chilappagari2007fault, varshney2011performance}. The objective of these works has been the storage of classical information reliably for long durations with little circuit complexity overhead.

It is natural to ask whether quantum information can be stored reliably using noisy components. The question is perhaps more pertinent in the quantum setting because, unlike the classical case where noise is due to defects, practical limitations, or resource constraints, noise in quantum systems is more fundamental---noise occurs due to interaction with the environment causing quantum states to decohere rapidly over time. Classical information storage capacity is measured in (logical) bits normalized by the number of binary registers and logic gates. However, one can store either classical information or quantum bits (qubits) in quantum states. In this paper, we consider the latter, which enables a memory system to store entangled qubits across multiple systems. Similar to the classical setting, we define quantum storage capacity as the number of logical qubits normalized by the total number of physical qubits (data qubits and ancillas), quantum gates, and measurements.

The idea of a quantum memory based on quantum error correction was first proposed as a 9-qubit code to store one logical qubit \cite{shor1995scheme}. A quantum memory system using a two-to-four qubit encoding scheme is proposed in \cite{gingrich2003all}, for its use as a repeater to reduce photon loss in quantum communication. Further, \cite{terhal2015quantum} provides a review of qubit and subsystem stabilizer codes, and their potential application in protecting quantum information in quantum memories.

There are currently several contenders for physically realizing quantum computing/memory systems including superconducting, ion-trap, quantum dot, and optical systems. In certain quantum memory systems, quantum states are mapped between photons and atomic systems, allowing effective storage of quantum information \cite{liu2001observation, phillips2001storage, schori2002recording}. Alternatively, the concept of cyclic quantum memory proposed in \cite{gingrich2003all} constitutes a delay line loop, complemented by a small linear optical quantum computing circuit, to repeatedly execute a quantum error-correction code (QECC) and preserve quantum coherence.

Here we focus on a general model of a quantum memory system, whose reliability is ensured using quantum error correction codes and correcting circuits. We prove that the achievable storage rate---defined as the rate of logical qubits stored normalized by circuit complexity rather than the number of channel uses---is a strictly positive constant. This achievability argument uses a class of quantum low-density parity-check (LDPC) codes known as quantum expander codes. These codes are of particular interest due to their inherent robustness in combating errors and providing efficient error correction, with little circuit complexity in their decoders. {\color{black} Moreover, we reduce the reliable storage problem to reliable communication of quantum information to obtain converse upper bounds on storage capacity. Inspired from works on classical reliable memories, we extend an entropy dissipation argument to the quantum setting to tighten the upper bounds further. Thus, we connect the fields of circuit complexity and information theory for the first time in the quantum setting.

Note that the achievable storage rate being strictly positive is not obvious because most quantum error correcting codes such as surface codes\cite{kitaev2003fault} and hypergraph product codes  \cite{tillich2013quantum} lack either an efficient decoder or a constant code rate. Additionally, fault-tolerant quantum computation (FTQC) is focused on proving the possibility of reliable quantum computation with constant overhead \cite{gottesman2013fault, fawzi2020constant}, but the emphasis is usually not on the exact value of this constant. Here, we take a step further---we are interested in the value of achievable storage rate (lower bound on storage capacity), in the context of reliable storage of logical qubits by also taking measurements into account towards circuit complexity. On the other hand, converse arguments to obtain upper bounds on quantum storage capacity have not been studied before as far as we know. Therefore, here we formalize the notion of entropy dissipation and noise in measurements, to obtain tighter converse bounds by extending the entropy dissipation argument in \cite{varshney2015toward} to the quantum setting.}

The remainder of the paper is organized as follows. Section~\ref{sec:definitions} provides a formal definition of a stable quantum memory system.  Section~\ref{sec:sys_model} presents the quantum memory system and noise models. Section~\ref{sec:background} provides a brief overview of quantum expander codes and bivariate bicycle LDPC codes. Section~\ref{sec:achievability} analyzes the achievability of stable memories and Section~\ref{sec:converse} gives converse bounds on storage capacity. Section~\ref{sec:time_complex_classical} provides a tighter non-asymptotic upper bound by considering the computational time complexity of the decoder. Section~\ref{sec:compare_lb_ub} compares the gap between the upper and lower bounds on storage capacities. Finally, Section~\ref{sec:conclusion} concludes the paper with limitations of the present work and potential future research directions.

\section{Reliable Memories}
Let us formalize notions of reliability, circuit complexity, and reliable memory systems.

\label{sec:definitions}
\begin{definition}
    \label{def:ckt_complexity}
    The complexity $\chi$ of a quantum memory is the total number of qubits (data qubits and ancillas), quantum gates, and measurements in the memory system.
\end{definition}

Measurements are usually not counted towards the quantum circuit complexity since measurements are typically performed at the end of the algorithm. However, in the context of quantum memories, measurements are performed periodically to protect the quantum state from decoherence. Therefore, in this paper, we count measurements towards circuit complexity.

\begin{definition}
    Storage overhead $\theta$ of a quantum memory is the ratio of the complexity of quantum memory constructed using unreliable components (qubits, gates, and measurements) to the number of logical qubits.
\end{definition}

\begin{definition}
    \label{def:asymp}
    Let $\{M_k\}$ be a sequence of quantum memories each with a storage capability of $k$ \textcolor{black}{logical} qubits. The quantum memory sequence is said to be \emph{stable} if it satisfies the following:
    \begin{enumerate}
        \item At time $t = 0$, each memory $M_k$ is initialized with a state $\rho_k = \mc{E}(\sigma_k)$, where $\sigma_k \in \mc{B}(\mc{H})$, $\dim(\mc{H}) = 2^k$ and $\mc{E}:\mc{B}(\mc{H}) \rightarrow \mc{B}(\mc{H}')$ is a CPTP map, which acts as an encoder.
        \item The complexity $\chi(M_k)$ of memory $M_k$ is bounded by $\theta k$, where $\theta$ is fixed for all $k$.
        \item For any $\epsilon>0$ and $T>0$, there exists a time $t \geq T$ and a CPTP map (a decoder) $\mc{D}:\mc{B}(\mc{H}') \rightarrow \mc{B}(\mc{H})$ such that $\underset{\sigma_k}{\min}~F(\mc{D}(\hat{\rho}_k), \sigma_k) \geq 1 - \epsilon$, for some $k$, {\color{black} where $\hat{\rho}_k$ is the state of $M_k$ at time $t$ given that its initial state at time $t=0$ was $\rho_k=\mc{E}(\sigma_k)$}.
    \end{enumerate}
    We assume that both the encoder $\mc{E}(\cdot)$ and the decoder $\mc{D}(\cdot)$ are ideal (noiseless), and are not a part of the memory system.
\end{definition}

\begin{definition}
    \label{def:storage_capacity}
    Quantum storage capacity $\mathfrak{Q}$ of a quantum memory is a number such that the storage overhead $\theta$ of any stable quantum memory sequence is bounded below by $1/\mathfrak{Q}$.
\end{definition}

Before proceeding further, let us recall a definition of the quantum capacity of a quantum channel. {\color{black}Note that unlike the quantum storage capacity $\mathfrak{Q}$ described in Definition~\ref{def:storage_capacity}, the quantum capacity $\mc{Q}$ does not consider circuit complexity into account.}

\begin{definition}
    \textit{Achievable communication rate and quantum capacity \cite{fawzi2022lower}:} For a qubit channel $\mathcal{N}:\mc{B}(\mc{H}) \rightarrow \mc{B}(\mc{H})$, a communication rate $R$ is called $\epsilon$-achievable if there exists $n_{\epsilon}$ such that for all $n \geq n_{\epsilon}$, there exists a sequence of encoders $\mc{E}_n:\mc{H}^{Rn} \rightarrow \mc{H}^{n}$ and decoders $\mc{D}_n:\mc{H}^{n} \rightarrow \mc{H}^{Rn}$ such that $F(\mc{D}_n \circ \mc{N}^{\otimes n} \circ \mc{E}_n(\sigma), \sigma) \geq 1-\epsilon$ for all $\sigma \in \mc{B}(\mc{H}^{Rn})$ and $\epsilon \in (0, 1]$. The quantum capacity $\mc{Q}(\mc{N})$ is the supremum of all achievable rates.
\end{definition}

\section{System Model}
\label{sec:sys_model}

In this section, we introduce a model of a quantum memory system inspired by \cite{taylor1968reliable}, and two noise models considered in this paper.

\subsection{Quantum Memory System Model}

\label{subsec:q_memory_sys_model}
\begin{figure*}
    \includegraphics[scale=0.57]{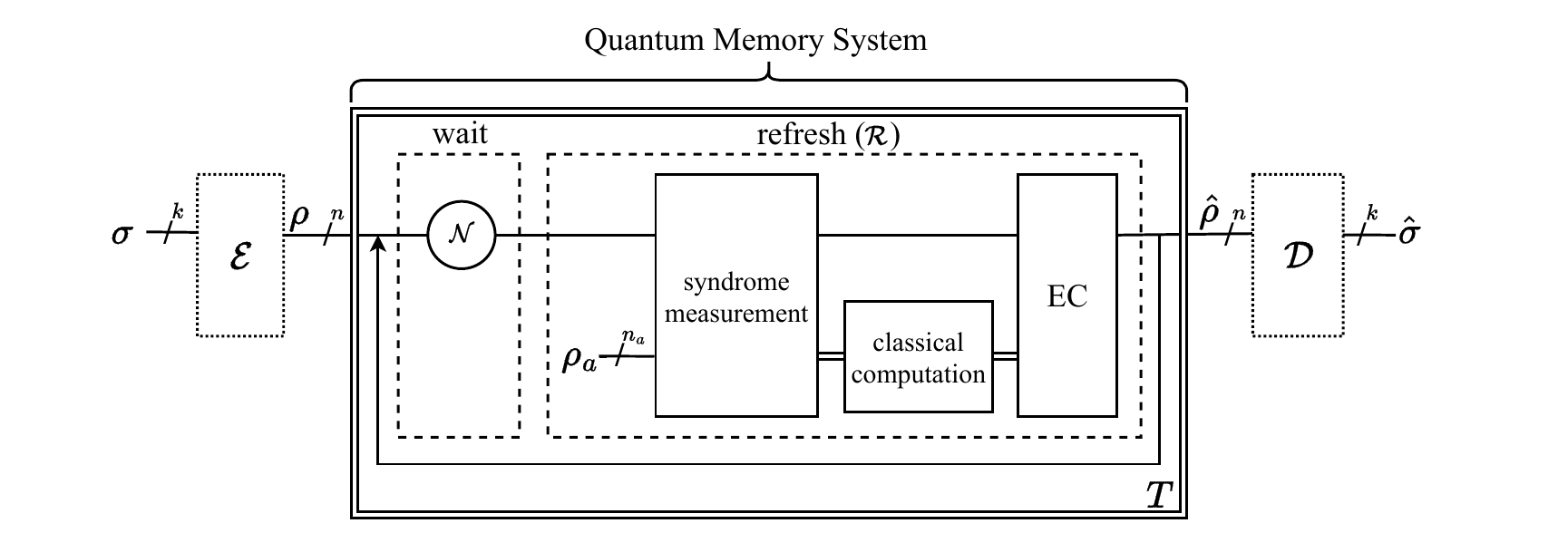}
    \caption{Schematic of a quantum memory system.}    \label{fig:reliable_quantum_memory_blk_diagram}
\end{figure*}

Figure~\ref{fig:reliable_quantum_memory_blk_diagram} shows a schematic of a quantum memory system. An encoder $\mc{E}:\mc{B}(\mc{H}) \rightarrow \mc{B}(\mc{H}'):\sigma \mapsto \rho$ is a CPTP map which encodes a $2^k$ dimensional quantum state $\sigma$ to a $2^n$ dimensional state $\rho$. A quantum memory system (enclosed in a double-lined box in Figure~\ref{fig:reliable_quantum_memory_blk_diagram}) takes the encoded state $\rho$ as input, which undergoes repeated wait-refresh cycles with a total duration $T$. A single wait-and-refresh is completed within a duration of time $\tau$. The (encoded) quantum state after time duration $T$ is indicated by $\hat{\rho}$, which is decoded by $\mc{D}$ to recover the quantum state $\hat{\sigma}$. During the wait duration, the qubits gradually decohere. The refresh block $\mc{R}$ contains the following sub-blocks: a syndrome measurement circuit which {\color{black} records} the error syndrome on the ancillas $\rho_a$, a classical decoder that obtains the estimate of the set of qubits requiring correction using the syndrome measurement, and an error correction (EC) sub-block that applies the correction (for example, bit-flip and phase-flip corrections) on noisy encoded qubits.

Let us clarify the system model further. In the first cycle, there is no wait duration, instead, the noise in physical qubits is due to the error in state preparation. We assume that the state preparation error is smaller than the error introduced during the wait duration in each subsequent cycle. Secondly, the loop from the output of the refresh block to the input of the wait block in Figure~\ref{fig:reliable_quantum_memory_blk_diagram} does not imply a physical feedback loop, because it does not represent a quantum evolution, and it could also violate the no-cloning theorem. Instead, it is similar to a template model in probabilistic graphical models. In fact, one can equivalently redraw the diagram by unrolling the circuit for $\lceil T/\tau \rceil$ time cycles. For example, in a superconducting quantum architecture, the refresh block corresponds to applying microwave pulses corresponding to the controlled unitary operations and measurements applied repeatedly in regular intervals on the same set of qubits, which does not involve a physical feedback loop. Some of the desirable properties of a reliable quantum memory system according to Definition~\ref{def:asymp} are as follows.
\begin{enumerate} 
    \item The residual physical noise after refresh should be bounded above by the {\color{black}error correction noise threshold} even when the syndrome measurements are noisy. This ensures the error is bounded over multiple cycles.
    \item There is an ability to create fresh ancilla qubits quickly at the beginning of every refresh step, and discard (or reset) them after error correction in each cycle (this is a necessary property of fault-tolerant quantum circuits).
    \item The refresh block should complete all operations (syndrome measurement and classical computing for error detection, and error correction) before {\color{black} the physical noise of the qubits exceeds the error correction noise threshold}.
\end{enumerate}

\subsection{Noise Model}

We consider two noise models, namely, local stochastic noise and depolarizing noise. Both noise models treat the corruption of quantum states as Pauli errors, i.e., bit flip, phase flip, and bit-phase flip errors.

\paragraph{Local Stochastic Noise}
Let $V$ be the set of data qubits (variable nodes) and $C$ be the set of ancillas (check nodes). {\color{black} An error $(E, D)$ is a random variable, where $E \subseteq V$ indicates the qubits with Pauli $X$ error\footnote{It suffices to consider one of $X$ or $Z$-type error due to symmetry \cite{fawzi2018efficient}.} and $D \subseteq C$ indicates the the syndrome measurement error, satisfying the local stochastic error model \cite{fawzi2020constant} with parameters $(p, q)$ if for all $S \subseteq V$ and $T \subseteq C$
}
    \begin{equation}
        \mathbb{P}[S \subseteq E, T \subseteq D] \leq p^{|S|} q^{|T|}.    
    \end{equation}
    In other words, the location of the errors is arbitrary but the probability of an error decays exponentially with its weight.

\begin{remark}
    The syndrome error parameter $q$ encapsulates all the following sources of errors, namely, errors in preparing the ancillas, noisy controlled gate operations involved in syndrome measurement, measurement errors, and gate errors in the correction step. This is bounded because the weights of the stabilizers are a constant for LDPC codes (will be discussed in Section~\ref{subsec:q_expander_codes}).
\end{remark}

\paragraph{Depolarizing Noise} 

The depolarizing noise with parameter $\tilde{p}$ acting on a qubit is described as follows \cite{khatri2020principles}: 
\begin{equation}
    \mathcal{N}_{\tilde{p}}(\rho) = (1-\tilde{p}) \rho + \frac{\tilde{p}}{3} (X \rho X + Y \rho Y + Z \rho Z).
\end{equation}
When a set of qubits is affected by depolarizing noise, unlike the local stochastic noise, we assume that the noise across qubits acts independently.

\section{Background on Quantum Error Correction}
\label{sec:background}
To be self-contained and establish notation, here we provide a brief overview of quantum error correction. A quantum error correction code maps a quantum state to a higher-dimensional subspace to protect it from decoherence by detecting (and correcting) noise-induced excursions from the code subspace. Fortunately, it is sufficient to correct only bit-flip and phase-flip errors to correct any arbitrary error in a quantum state \cite[Theorem~2]{gottesman2010introduction}. In this context, the stabilizer formalism provides a framework to construct quantum error correction codes with desired properties \cite{gottesman2010introduction}.

\subsection{Stabilizer Formalism, CSS, and LDPC Codes}
A stabilizer is a subset of $n$-qubit Pauli operators $\{\pm 1, \pm i\} \times \{I, X, Y, Z\}^n$ that forms an Abelian group. The corresponding quantum stabilizer code is defined as the +1 eigenspace of the stabilizer. An especially elegant stabilizer code is provided by a CSS (Calderbank, Shor, and Steane) construction using two classical codes $\mc{C}_X$ (to correct bit-flip errors) and $\mc{C}_Z$ (to correct phase-flip errors) obeying $\mc{C}^{\perp}_X \subset \mc{C}_Z$ (or $\mc{C}^{\perp}_Z
\subset \mc{C}_X$), which is equivalently stated using the corresponding parity check matrices as $H_X H_Z^T = 0$. Syndrome measurement is performed by measuring the eigenvalues of the stabilizer using entangling gates between the data qubits and ancillas. Error correction involves  applying a set of X and/or Z gates corresponding to the error pattern recovered from the syndrome.

Quantum CSS codes constructed using a pair of classical LDPC codes (quantum LDPC codes) are practically interesting because the stabilizers have constant weight and each qubit is part of a constant number of stabilizers, enabling computationally efficient decoders. A particular instance from the family of quantum LDPC codes is a quantum expander code proposed in \cite{leverrier2015quantum}, described as follows.

\subsection{Quantum Expander Codes}
\label{subsec:q_expander_codes}
A quantum expander code is constructed as a hypergraph product of the Tanner graph of two classical expander codes.\footnote{We assume the two classical codes are identical.} Specifically, let $\mc{C}$ be a classical expander code with Tanner graph $G = \{A \cup B,~\mc{E}\}$ which is a $(\gamma_A, \delta_A, \gamma_B, \delta_B)$-expanding graph with constant left and right degrees, $d_A$ and $d_B$, respectively. Denote the cardinalities $|A| = n_A$ and $|B| = n_B$. The hypergraph product code of $\mc{C}$ with itself creates a new Tanner graph with $(n_A + n_B)^2$ vertices, where the set of qubits (variable nodes) are indexed by $A^2 \cup B^2$, the $X$-type and $Z$-type generators (check nodes) are indexed by $A \times B$ and $B \times A$, respectively. The hypergraph product induces parity check matrices, $H_X = (I_{n_A} \otimes H,~H^T \times I_{n_B})$ and $H_Z = (H \otimes I_{n_A},~I_{n_B} \times H^T)$ such that $H_X H_Z^T = 0$, validating that it is a CSS code. The generators of the quantum expander code have a constant weight $d_A + d_B$, and it is LDPC with parameters $[[n, k, d_{\min}]]$, where $n = n_A^2 + n_B^2$, $k \geq (n_A - n_B)^2$, and $d_{\min} = d_{\min}(\mc{C})$.

Quantum expander codes have constant asymptotic rate, essential to achieving a strictly positive asymptotic storage rate. Further, an efficient small set-flip decoding algorithm runs in $O(\log n)$ time \cite{fawzi2018efficient}.

\subsection{Bivariate Bicycle Quantum LDPC Codes}

Bivariate bicycle (BB) quantum LDPC codes \cite{bravyi2024high}, provide a high-rate alternative to the surface codes in the non-asymptotic (in the number of physical qubits) setting. The construction of a BB quantum LDPC code is based on $\ell \times \ell$ cyclic shift matrices $S_\ell$ whose $i$th row has only one non-zero entry at the column $i+1 ~(\text{mod } \ell)$. Denote $x = S_\ell \otimes I_m$ and $y = I_\ell \otimes S_m$. A BB code is described by the pair $(A, B)$ where $A$ and $B$ are of the form $\sum_{i=1}^{3} z_i^{\zeta_i}$, with $z_i \in \{x, y\}$, and $\zeta_i \in \{1, \ldots, \ell-1\}$ if $\zeta_i = x$, or $\zeta_i \in \{1, \ldots, m-1\}$ if $\zeta_i = y$. A BB quantum LDPC code is defined by the following $\ell m \times 2\ell m$ parity check matrices: $H_X = (A,B)$ and $H_Z = (B^T,A^T)$. The generators of BB quantum LDPC code have a constant weight of $6$ with parameters $[[n, k, d_{\min}]]$, where $n = 2\ell m$, $k = 2\dim(\ker(A) \cap \ker(B))$ and $d_{\min} = \min\{|v| : v\in \ker(H_X)\setminus \text{rs}(H_Z)\}$. Some examples of high-distance BB quantum LDPC codes found by numerical search are listed in \cite[Table~3]{bravyi2024high}.

The Tanner graph associated with the BB quantum LDPC codes has thickness $\leq 2$, which makes it a practically attractive choice for chip-based quantum circuits. However, unlike the quantum expander codes discussed earlier, it is unclear whether the asymptotic rate of the BB quantum LDPC code is strictly positive.

\section{Lower bounds (Achievability)}

\label{sec:achievability}

In this section, we try to find a lower bound on the storage capacity of a quantum {\color{black} memory}. In particular, we find a range of physical noise (local stochastic and depolarizing noise) levels for which the probability of error can be ensured to be arbitrarily small with the size of the quantum memory, using a specific quantum error correcting code. The ratio of the number of logical qubits to the total number of physical qubits (including ancillas), quantum gates, and measurements gives the lower bound on storage capacity.

{\color{black} In quantum expander codes, since each check (variable) node is connected to a constant number of variable (check) nodes, syndrome measurement and error correction can be completed in constant (in the number of physical data qubits $n$) time privided they are performed in parallel. Moreover, it is a standard assumption in FTQC that classical computation is fast and error-free \cite{gottesman2010introduction, gottesman2013fault, fawzi2020constant}. Therefore unless otherwise specified, we assume that the operations associated with the refresh block have a constant time complexity  with respect to $n$
. In this context, we assume that each wait and refresh cycle is completed within time duration $\tau$ constant with respect to $n$. Therefore, until Section~\ref{sec:time_complex_classical}, where we discuss the effect of classical computation time on storage capacity, we omit the argument $\tau$ while referring to the physical noise parameter, i.e., $p := p(\tau)$.}

\begin{table}
    \centering
    \vspace{2mm}
    \begin{tabular}{|c|l|}
         \hline
         \textbf{Parameter} & \multicolumn{1}{c}{\textbf{Description}} \vline  \\
         \hline
         $p$ & physical noise parameter of data qubits \\
         \hline
         $q$ & syndrome measurement noise parameter\\
         \hline
         $p_r$ & residual noise parameter of data qubits\\         
          & after refresh\\         
         \hline
         $p_{th}$ & error correction noise threshold\\
         \hline
         $P_e$ & probability of logical error \\
         \hline
    \end{tabular}
    \caption{Notations and their description.}
    \label{tbl:notations}
\end{table}

\begin{lemma}
    \label{lemma:pr_plus_p}
    Let $E_r$ be the {\color{black} residual} error after refresh and $E_w$ be the error introduced during the wait duration, both of which are local stochastic with parameters $p_r$ and $p$, respectively. The total error experienced by the quantum state at the input of every refresh operation is $(E_r \cup E_w, D)$ which is local stochastic with parameters $(p_r+p, q)$.
\end{lemma}
\begin{proof}
    The residual error $E_{r}$ right after error correction is local stochastic with parameter $p_r$. Then, waiting for a certain duration of time before the next round of refresh incurs additional error $E_w$ which is local stochastic with parameter $p$. The total error before the second round of refresh is $E_{r} \cup E_w$, which is again local stochastic with parameter $p_r + p$ (see \cite[Lemma~11]{tomamichel2019quantum} for the proof that composition of local stochastic errors is also local stochastic). Ancillas are freshly prepared for capturing the syndrome measurement in each cycle. Therefore, the error in syndrome measurement $D$ is independent of the error in data qubits, and is local stochastic with parameter $q$.
\end{proof}

The following lemma provides a relationship between the probability of logical error and fidelity between the quantum states, which is required to prove Theorem~\ref{thm:qexp_storg_cap_lb}.

\begin{lemma}
    \label{lemma:prob_err_fidelity}
    Suppose a quantum state $\rho$ is corrupted by Pauli errors, and $\sigma$ is the quantum state conditioned on the occurrence of a logical error. Then, for any $\epsilon \in [0, 1]$, if the probability of logical error $P_e \leq \epsilon$, the fidelity is bounded as $F(\rho, \rho') \geq 1-\epsilon$, where $\rho' = (1-P_e) \rho + P_e \sigma$.
\end{lemma}
\begin{proof}
    The fidelity between $\rho$ and $\rho'$ is bounded as follows:
    \begin{equation}
        \begin{split}
            &F(\rho', \rho) = \Tr \lt( \sqrt{\sqrt{\rho} \rho' \sqrt{\rho}} \rt)^2 \\
            & = \Tr \lt( \sqrt{ (1-P_e)\sqrt{\rho} \rho \sqrt{\rho} + P_e \sqrt{\rho} \sigma \sqrt{\rho}  } \rt)^2 \\
            & \geq \Tr \lt( \sqrt{ (1-P_e)\sqrt{\rho} \rho \sqrt{\rho}} \rt)^2 \\
            & = \Tr \lt( \sqrt{ (1-P_e)} \rho \rt)^2 = 1-P_e \geq 1-\epsilon.
        \end{split}
    \end{equation}
\end{proof}

To determine an achievable storage rate, we shall consider a quantum memory system protected using a quantum expander code proposed in \cite{fawzi2018efficient}, which is constructed as a hypergraph product of two classical expander codes\footnote{We assume that the two classical codes are identical.}. The following lemma provides a bound on the logical probability of error decaying exponentially with the number of physical qubits provided the physical error is less than a threshold.

\begin{lemma}
For a quantum memory system protected using a quantum expander code described in \ref{subsec:q_expander_codes} against local stochastic noise with parameters $(p+p_r, q)$, the logical probability of error after one round of refresh is given by:
\begin{equation}
    P_e(n) \leq C n \lt( \frac{p + p_r}{p_{th}} \rt)^{C' \sqrt{n}}.
    \label{eqn:Pe_qexpander_qmemory}
\end{equation}
where $p_{th}$ is the fault-tolerant threshold, the residual error after error correction is local stochastic with parameter $p_r = K q^{1/c_0}$ satisfying $p+p_r < p_{th}$, $K$ is a constant independent of $p$ and $q$ (see \cite[proof of Theorem~3]{fawzi2020constant}), and $p_{th}$, $c_0$, $C$ and $C'$ are constants that depend on the parameters of quantum expander code (see \cite[Section~3.1]{fawzi2020constant} for details).
\end{lemma}
\begin{proof}
    From Lemma~\ref{lemma:pr_plus_p}, noise at the input of the refresh block is local stochastic with parameter at most $p+p_r$. If $p+p_r<p_{th}$, then from \cite[Theorem~1]{fawzi2018efficient} the logical probability of error is bounded above as \eqref{eqn:Pe_qexpander_qmemory}.
\end{proof}

Now, the following theorem provides a strictly positive lower bound on storage capacity for a quantum memory system protected using a quantum expander code. Specifically, consider a quantum memory system described in Section~\ref{subsec:q_memory_sys_model} protected using quantum expander code described in \cite{fawzi2018efficient}.

\begin{theorem}
    \label{thm:qexp_storg_cap_lb}
    Suppose the qubits and syndrome measurements experience local stochastic noise with parameter pair $(p, q)$. Let the residual noise after each cycle be local stochastic with parameter $p_r$. If $p+p_r < p_{th}$, then there exists a sequence of quantum memories with a strictly positive storage capacity $\mathfrak{Q}(\mc{N}_p)$. Specifically, the storage capacity of this memory system is bounded below as
    \begin{equation}
        \mathfrak{Q}(\mc{N}_p) \geq \frac{1}{\frac{3}{R} + (3 + d_A + d_B) \frac{2}{R_c} \lt(\frac{1}{R_c} - 1\rt)},
    \end{equation}
    where $R_c = \frac{n_A - n_B}{n_A} = 1-\frac{d_A}{d_B}$ and $R = \frac{(n_A - n_B)^2}{n_A^2 + n_B^2} = \frac{(d_B - d_A)^2}{d_B^2 + d_A^2}$ are the design rates of the classical expander code and the quantum expander code, respectively.
\end{theorem}

\begin{proof}
    For the quantum expander code proposed in \cite{fawzi2018efficient}, note that $k = (n_A - n_B)^2$ is the number of logical qubits, $n = n_A^2 + n_B^2$ is number of encoded data qubits (does not include ancillas), and $n_a = n-k = 2 n_A n_B$ is the number of ancillas. Denote $n_H = 2 n_A n_B $ as the number of Hadamard gates required for bit-flip error measurement, $n_{synd} = 2 n_A n_B (d_A + d_B)$ as the number of C-Z and C-X gates operating on encoded data qubits for syndrome measurement with ancillas as control and encoded qubits as target, $n_m = 2 n_A n_B$ as the number of Z-basis measurements, and $n_{EC} = 2 (n_A^2 + n_B^2)$ as the number of classically controlled X and Z gates for error correction. Then the total number of gates for error correction is
    \begin{equation}        
    \begin{split}
        \chi_g &= n_H + n_{synd} + n_{EC} \\
        &= 2 n_A n_B + 2 n_A n_B (d_A + d_B) + 2 (n_A^2 + n_B^2).
    \end{split}
    \end{equation}
    The achievable quantum storage rate is the ratio of the number of logical qubits to the complexity of the memory system:
    \begin{equation}
    \begin{split}    
        R_s &= \frac{k}{\chi} = \frac{k}{n + n_a + \chi_g + n_m} \\
        &= \frac{1}{3 \frac{n_A^2 + n_B^2}{(n_A - n_B)^2} + 2 (3 + d_A + d_B) \frac{n_A n_B}{(n_A - n_B)^2}}.
    \end{split}
    \end{equation}
    Observing that the design rate of classical expander code is $R_c = \frac{k_c}{n_c} = \frac{n_A - n_B}{n_A}$ and the design rate of quantum expander code is $R = \frac{(n_A - n_B)^2}{n_A^2 + n_B^2}$, the achievable storage rate can be rewritten as
    \begin{equation}
        R_{s} = \frac{1}{\frac{3}{R} + (3 + d_A + d_B) \frac{2}{R_c} \lt(\frac{1}{R_c} - 1\rt)}. \label{eqn:storage_capacity_lb}
    \end{equation}
    Note that the logical error rate after $L = \lceil T/\tau \rceil$ rounds of wait-refresh cycles can be obtained by applying the union bound:
\begin{equation}
\begin{split}
    P^{(T)}_e(n) &\leq  \sum\limits_{l=1}^{L} P^{(l)}_e(n) \leq L P_e(n)\\
    & = C \left\lceil\frac{T}{\tau}\right\rceil n \lt( \frac{p + p_r}{p_{th}} \rt)^{C' \sqrt{n}},
\end{split}
\end{equation}
where the last expression is obtained by substituting for $P_e(n)$ in \eqref{eqn:Pe_qexpander_qmemory}. For any given $T>0$, if $p+p_r < p_{th}$, then $P_e^{(T)}(n)$ can be made arbitrarily small for large values of $n$ (i.e., $P_e^{(T)}(n) \leq \epsilon$ for any $\epsilon \in (0, 1]$). Consequently, from Lemma~\ref{lemma:prob_err_fidelity}, the fidelity $F(\rho, \hat{\rho}) \geq 1-\epsilon$ can be achieved. Therefore, the quantum storage capacity can be bounded as $\mathfrak{Q}(\mc{N}_p) \geq R_s$.
\end{proof}

\begin{remark}
    Note that the storage rate is a (strictly) positive constant, and does not vanish with increasing $n$.
\end{remark}

\begin{corollary}
    Consider a quantum memory system described in Theorem~\ref{thm:qexp_storg_cap_lb}, but the qubits experience i.i.d. depolarizing noise $\mc{N}_{\tilde{p}}$ with parameter $\tilde{p} = 3p/2$ and the refresh block is noiseless. If $p < p_{th}$, then the storage capacity of this memory system is bounded below as $\mathfrak{Q}(\mc{N}_{\tilde{p}}) \geq R_s$.
\end{corollary}

\begin{proof}
    Since the refresh block is noiseless, $p_r = 0$. Therefore, if the qubits experience a local stochastic noise with parameter $p$ with $p<p_{th}$, then the logical error can be made arbitrarily small for a sufficiently large $n$. Suppose the logical error rate is bounded by $\epsilon$, then from \cite[Lemma~29]{fawzi2018efficient} the logical error rate with i.i.d. depolarizing noise with parameter $\tilde{p} = 3p/2$ is also bounded by $\epsilon$. Therefore, the quantum storage capacity of the memory system with i.i.d. depolarizing noise is also bounded below by $R_s$ in \eqref{eqn:storage_capacity_lb}.
\end{proof}

\section{Upper bounds (Converse)}
\label{sec:converse}

In this section, we derive upper bounds on storage capacity in both asymptotic and non-asymptotic scenarios, where the latter considers fixed number of physical qubits.

\subsection{Asymptotic bounds}
\label{subsec:asymp_bounds}
The following theorem provides an asymptotic upper bound on storage capacity when a quantum memory system is corrupted by an i.i.d. noise across qubits (using an argument similar to \cite[proof of Lemma~1]{uthirakalyani2023converse}).
\begin{theorem}
    \label{thm:converse_asymp}
    Consider a quantum memory system where each qubit is corrupted independently by a qubit channel $\mc{N}$, and let $\mc{Q}(\mc{N})$ denote the corresponding quantum capacity. Then the  storage capacity is bounded as $\mathfrak{Q}(\mc{N}) \leq \mc{Q}(\mc{N})$.
\end{theorem}
\begin{proof}
    Consider a quantum memory system with quantum information stored across $n$ physical qubits, where $n$ is arbitrary. Let $\mc{N}^{\otimes n}$ and $\mc{R}$ indicate the CPTP maps corresponding to the wait duration and refresh, respectively. Then the decoded quantum state after $L$ cycles of wait and refresh is $\hat{\sigma} = \mc{D} \circ (\mc{R} \circ \mc{N}^{\otimes n})^{\circ L} \circ \mc{E}(\sigma)$. Note that $\hat{\sigma}$ can be expressed as $\mc{D}' \circ \mc{N}^{\otimes n} \circ \mc{E}(\sigma)$ for some $\mc{D}'$, in particular for $\mc{D}' = \mc{D} \circ (\mc{R} \circ \mc{N}^{\otimes n})^{\circ L-1} \circ \mc{R}$. This implies that there exists $\mc{D}'$ with $\sigma' = \mc{D}' \circ \mc{N}^{\otimes n} \circ \mc{E}(\sigma)$ such that $F(\sigma, \sigma') \geq F(\sigma, \hat{\sigma})$. Therefore, $F(\sigma, \sigma') \geq 1-\epsilon$ is a necessary condition for $F(\sigma, \hat{\sigma}) \geq 1-\epsilon$. In other words, the necessary condition for reliable storage of quantum state $\sigma$ in a memory system with each qubit corrupted by $\mc{N}$ is the ability to reliably communicate $\sigma$ through the independent channel $\mc{N}^{\otimes n}$. Storing quantum information with a rate greater than $\mc{Q}(\mc{N})$ would violate condition (3) of Definition~\ref{def:asymp}. Therefore, the storage capacity is bounded as $\mathfrak{Q}(\mc{N}) \leq \mc{Q}(\mc{N})$.
\end{proof}

The upper bound in Theorem~\ref{thm:converse_asymp} is loose because it takes noise in only one layer into account and the decoder complexity is not considered in deriving the bound. The next theorem (Theorem~\ref{thm:converse_asymp_dissip}) tightens this bound by incorporating a lower bound on the decoder complexity using an entropy dissipation argument similar to \cite{varshney2015toward}. Of separate interest, we also refine the upper bound on the storage capacity of classical memories of \cite{varshney2015toward} in Appendix \ref{apndx:classical_storage_capacity_ub}.

{\color{black} The effect of a noisy channel on data qubits can be equivalently described as a joint unitary evolution between the data qubits and the environment, followed by a discarding of the environment, which manifests as decoherence (entropy accretion) in the data qubits. The goal of the refresh map $\mc{R}$ is to project the quantum state to the {\color{black}codeword subspace}, disentangling the data qubits from the environment. One of the ways to achieve this is by using a set of orthogonal projectors as in stabilizer codes (measuring the ancillas in computational basis) followed by error detection and error correction, {\color{black}or alternatively one can consider a more general channel decoding approach of applying a coherent quantum instrument}, followed by a suitable isometry (as indicated in \cite[(24.28) and (24.50)]{wilde2019classical}) and re-encoding.  {\color{black}This general decoding method is depicted in Figure~\ref{fig:reliable_quantum_memory_blk_diagram} as a syndrome projective measurement, followed by classical post-processing and error correction.} To obtain a tighter converse bound on storage capacity in Theorem~\ref{thm:converse_asymp_dissip}, we consider the entropy accreted in the data qubits being dissipated using {\color{black}a collection of suitable projective measurements.} To this end, we start by defining complete entropy dissipation as follows.

\begin{definition}
    \label{def:complete_entropy_dissip}
    \textbf{Complete entropy dissipation:} The refresh map $\mc{R}$ dissipates entropy completely if and only if {\color{black}for a given $\varphi_{RM}$ and }for any $\delta \in (0, 1)$,\\
    \begin{equation}
        F(\mc{R}_{M} \circ \mc{N}^{\otimes n}_{M}(\varphi_{R M}), \varphi_{R M}) \geq 1-\delta.
        \label{eqn:complete_entropy_dissip}
    \end{equation}
    where $R$ is the purifying system or a reference system to which $M$ is entangled.\footnote{We omit the subscript $k$ for quantum memory $M_k$ to simplify notation.}
\end{definition}
In other words, $\mc{R}$ is a complete entropy dissipator if and only if the resulting state after one cycle of wait-and-refresh $\mc{R}_{M} \circ \mc{N}^{\otimes n}_{M}$ matches the initial state with high fidelity. This happens only when $\mc{R}_{M} \circ \mc{N}^{\otimes n}(\varphi_{R M})$ is (almost) completely disentangled from the environment. The following lemma provides a lower bound on the entropy required to be dissipated (entropy accretion) by $\mc{R}$ in every cycle.

\begin{lemma}
    \label{lemma:entropy_dissip_min}
    A complete entropy dissipator $\mc{R}$ associated with a quantum memory $M$ storing $k \leq I^c(\mc{N}^{\otimes n})$ logical qubits must dissipate an entropy of at least $H_k^{acc}(\mc{N}^{\otimes n})$ in every cycle, where $H_k^{acc} = \min_{\varphi_{RM}} H(RM)_{\mc{N}_M^{\otimes n}}(\varphi_{RM})$ such that $I^c(R \rangle M)_{\mc{N}_M^{\otimes n}(\varphi_{RM})} \geq k$.
\end{lemma}
\begin{proof}
{\color{black} Let $\mc{R}$ be a complete entropy dissipator such that \eqref{eqn:complete_entropy_dissip} holds for some $\varphi_{RM} = \ketbra{\varphi}_{RM}$ satisfying $I^c(R \rangle M)_{\mc{N}_M^{\otimes n}(\varphi_{RM})} \geq k$.} Suppose $V_{ME}$ is a Stinespring dilation of $\mc{N}_M^{\otimes n}$, then 
\begin{align}
    \mc{N}_M^{\otimes n}&(\varphi_{RM}) = \Tr_E{\lt(V_{ME} \ketbra{\varphi}_{RM} \otimes \ketbra{e_0}_E V^\dagger_{ME}\rt)} \\
    &= \Tr_E{\lt(\sum_{i=1}^r \alpha_i \ket{\varphi_i}_{RM} \ket{\tilde{e}_i}_E \sum_{j=1}^r \alpha^*_j \bra{\varphi_j}_{RM} \bra{\tilde{e}_j}_E\rt)} \\
    &= \sum_{i=1}^r |\alpha_i|^2 \ketbra{\varphi_i}_{RM}, \label{eqn:entropy_dissip_stinespring}
\end{align}
where $\{\ket{\varphi_i}_{RM}\}_i$ and $\{\ket{\tilde{e}_i}_E\}_i$ are orthogonal pure states with $\sum_{i=1}^r |\alpha_i|^2 = 1$ and Schmidt rank $r$. From Definition~\ref{def:complete_entropy_dissip}, $F\lt( \mc{R}_M (\sum_{i=1}^r |\alpha_i|^2 \ketbra{\varphi_i}_{RM}),~\ket{\varphi}_{RM} \rt) \geq 1-\delta$ holds. Note that $\sum_{i=1}^r |\alpha_i|^2 \ketbra{\varphi_i}_{RM}$ is a classical mixture of orthogonal pure states. Since $\ket{\varphi}_{RM}$ is a pure state, for the fidelity to be arbitrarily close to 1, $\mc{R}$ needs to dissipate the entropy of at least $H(RM)_{\mc{N}_M^{\otimes n}}(\varphi_{RM}) = H(|\alpha_1|^2, \ldots, |\alpha_r|^2) = -\sum_{i=1}^r |\alpha_i|^2 \log_2(|\alpha_i|^2)$. Minimizing $H(RM)_{\mc{N}_M^{\otimes n}(\varphi_{RM})}$ over $\varphi_{RM}$ gives the minimum entropy required to be dissipated by $\mc{R}$.
\end{proof}

In general, determining the lower bound on entropy accretion can be challenging. However for unitarily covariant channels, obtaining the entropy accretion bound is simpler, which we show in the following corollary to Lemma~\ref{lemma:entropy_dissip_min}.
{\color{black}
\begin{corollary}
    \label{cor:entropy_dissip_min_unitary_covariant}
    If $\mc{N}$ is a unitarily covariant channel, then the entropy accretion is bounded below as $H^{acc}(\mc{N}^{\otimes n}) \geq n - I^c(\mc{N}^{\otimes n})$.
\end{corollary}
\begin{proof}
    {\color{black} The $n^{th}$ extension $\mc{N}^{\otimes n}$ of a unitarily covariant channel $\mc{N}$ is also unitarily covariant.} Then by definition we have $\mc{N}_M^{\otimes n}(\varphi_{RM}) = U_M^\dagger \mc{N}_M^{\otimes n}(U_M \varphi_{RM} U_M^\dagger) U_M$, whose dilation representation is $\mbox{Tr}_E\lt(V_{ME} \lt(\varphi_{RM} \otimes \ketbra{e}\rt) V_{ME}^\dagger \rt) = U_M^\dagger \mbox{Tr}_E\lt( V_{ME} (U_M \varphi_{RM} U_M^\dagger \otimes \ketbra{e}) V_{ME}^\dagger \rt) U_M$, which implies $H(RM)_{\mc{N}^{\otimes n}(\varphi_{RM})} = H(RM)_{\mc{N}^{\otimes n}(U_M \varphi_{RM} U_M^\dagger)}$ for any unitary $U$. Therefore, $H(RM)_{\mc{N}^{\otimes n}(\varphi_{RM})}$ is a constant for any $\varphi_{RM}$. Substituting $\varphi_{RM} = \Phi^+$ and maximizing the coherent information term with respect to $\varphi_{RM}$ we obtain the lower bound on entropy accretion as follows
    \begin{align}
        H_k^{acc}(\mc{N}^{\otimes n}) &= H(RM)_{\mc{N}_M^{\otimes n}(\varphi_{RM})} \\
        &= H(M)_{\mc{N}_M^{\otimes n}(\varphi_{RM})} - I^c(R \rangle M)_{\mc{N}_M^{\otimes n}(\varphi_{RM})} \\
        &\geq n - I^c(\mc{N}^{\otimes n}),
    \end{align}
    for any $k \leq I^c(\mc{N}^{\otimes n})$, where the last inequality is due to definition of channel coherent information $I^c(\mc{N}^{\otimes n}) = \underset{\varphi_{RM}}{\max}~I^c(R \rangle M)_{\mc{N}_M^{\otimes n}(\varphi_{RM})}$.
\end{proof}
}

Now, the following lemma provides a {\color{black}hypercontractivity} condition required to obtain a tighter upper bound on storage capacity in Theorem~\ref{thm:converse_asymp_dissip}.

\begin{lemma}
\label{lemma:strong_dpi}
Consider a unitarily covariant channel $\mc{N}$ satisfying the following conditions:
\begin{enumerate}[label=C\arabic*.]
    \item If $I_{n, \lambda}^c(\mc{R}) = I_{n, \lambda-1}^c(\mc{R})$, then $F(\rho_k^{(\lambda-1)}, \mc{R} \circ \mc{N}_M^{\otimes n} \circ \rho_k^{(\lambda-1)}) \geq 1-\delta$.
    \item For any CPTP map $\mc{R}$, only one of the following holds:
    \begin{enumerate}[label=\roman*., ref=\roman*]
        \item $I_{n, \lambda}^{c}(\mc{R}) = I_{n, \lambda-1}^{c}(\mc{R})$.
        \item $I_{n, \lambda}^{c}(\mc{R}) \leq \eta I_{n, \lambda-1}^{c}(\mc{R}) \text{ for some } \eta \in (0,1)$.
    \end{enumerate}        
\end{enumerate}
for all $\lambda \in \{1, \ldots, L\}$ and $\eta$ is constant with respect to $n$ and $L$. Here $I^{c}(\cdot)$ is the channel coherent information and $I^{c}_{n, \lambda}(\mc{R}) = I^{c}(\mc{N}^{\otimes n} \circ (\mc{R} \circ \mc{N}^{\otimes n})^{\circ \lambda})$. A quantum memory system with $n$ physical qubits can store quantum information reliably at a positive rate only if $\mc{R}$ dissipates the accreted entropy completely in every cycle.
\end{lemma}
\begin{proof}    
    Suppose $\mc{R}$ does not disspipate entropy completely, i.e., $F(\rho_k^{(\lambda-1)}, \mc{R} \circ \mc{N}_M^{\otimes n} \circ \rho_k^{(\lambda-1)}) < 1-\delta$. Then from $(C1)$, we have $I_{n, \lambda}^c(\mc{R}) \neq I_{n, \lambda-1}^c(\mc{R})$, which along with data processing inequality of coherent information implies that condition $(C2.ii)$ holds. Applying condition $(C2.ii)$ recursively, the rate at which quantum information can be stored reliably for $L$ cycles is bounded above by $(\eta)^{L-1} I^{c}(\mc{N}^{\otimes n})$ which vanishes for large $L$. Therefore, the refresh map $\mc{R}$ being a complete entropy dissipator is a necessay condition for storing quantum information reliably forever (for arbitrary $L$) at a positive rate.
\end{proof}

{\color{black} Note that in each cycle the refresh block dissipates the accreted entropy by performing a syndrome measurement followed by classical post-processing and error correction as shown in Figure~\ref{fig:reliable_quantum_memory_blk_diagram}. A syndrome measurement entails performing a set of stabilizer measurements, which is in turn a collection of projective measurements on the set of ancillas entangled with the data qubits. One can interpret that each projective measurement collapses the system's state progressively into more and more definite states, gradually dissipating the entropy accumulated in the system. For example, in the quantum expander code in Section~\ref{sec:achievability}, the stabilizer measurements contain projective measurements that are the $X$ and $Z$-basis measurements performed on the ancilla qubits.

The following theorem provides an upper bound on quantum storage capacity assuming that we are equipped with multiple instances of a projective measurement. The idea is to take into account the number of projective measurements required to dissipate the entropy accumulated in each cycle towards the quantum circuit complexity. Note that we do not account for the gate complexity of a projective measurement in the converse bounds, i.e., it has a unit contribution towards the circuit complexity.}

\begin{theorem} 
\label{thm:converse_asymp_dissip}
Consider a quantum memory equipped with {\color{black}a set of projectors} $\mc{P} = \{\Pi_u\}_u$ satisfying $\sum_u \Pi_u = \mathbbm{1}$. Suppose each data qubit is corrupted independently by a unitarily covariant qubit channel $\mc{N}$ satisfying conditions $C1$ and $C2$ in Lemma~\ref{lemma:strong_dpi}, then the quantum storage capacity is bounded as
\begin{equation}
    \mathfrak{Q}(\mc{N}) \leq \frac{\mc{Q}(\mc{N}) H^{dis}(\mc{P})}{1 + H^{dis}(\mc{P}) - \mc{Q}(\mc{N})},
    \label{eqn:converse_asymp_dissip}
\end{equation}
where $H^{dis}(\mc{P}) = \sup_{\psi} \lt( H(\psi) - \sum_{u} p_{\mathbf{u}}(u) H(\psi_u) \rt)$, $H(\psi)$ is the von Neumann entropy of density operator $\psi$, $p_{\mathbf{u}}(u) = \Tr(\Pi_u \psi)$ is the probability that the measurement outcome is $u$ and $\psi_u = \frac{\Pi_u \psi \Pi_u}{p_{\mathbf{u}}(u)}$ is the post-measurement state corresponding to the projector $\Pi_u$.
\end{theorem}

\begin{proof}
If the circuit complexity of an $n$-qubit memory system is $\mc{R}$ is $\chi$, then the storage rate is bounded as $\frac{I^c(\mc{N}^{\otimes n})}{n+\chi}$. To find a lower bound on $\chi$, observe (from Corollary~\ref{cor:entropy_dissip_min_unitary_covariant}) that at the end of the wait duration (just before the next round of refresh) the entropy build-up in the $n$-qubit quantum memory is at least $n-I^{c}(\mc{N}^{\otimes n})$. From Lemma~\ref{lemma:strong_dpi}, for quantum information to be reliably recovered after an arbitrary time duration (or arbitrary number of cycles $L$), the entropy built-up in every cycle must be completely dissipated. One way to dissipate the accreted entropy is through a set of projective measurements {\color{black}(for example, stabilizer measurements)}. The maximum entropy dissipated by the projective measurement $\mc{P}$ is given by $H^{dis}(\mc{P}) = \sup_{\psi} \lt( H(\psi) - \sum_{u} p_{\mathbf{u}}(u) H(\psi_u) \rt)$, where the first term and the second term are the entropies before and after measurement, respectively, $\psi_u$ is the post measurement state corresponding to the projector $\Pi_u$, and $p_{\mathbf{u}}(u) = \Tr(\Pi_u \psi)$ is the probability that the post measurement state is $\psi_u$. The supremum is over the set of all density operators $\psi$ (acting on Hilbert space of arbitrary dimension). Consequently, the number of measurements required to completely dissipate the accreted entropy from the data qubits is at least $\frac{n-I^{c}(\mc{N}^{\otimes n})}{H^{dis}(\mc{P})}$, which is a lower bound on the circuit complexity $\chi$. Therefore, the storage capacity is bounded as $\mathfrak{Q}(\mc{N}) \leq \frac{I^{c}(\mc{N}^{\otimes n})}{n+\chi} \leq \frac{I^{c}(\mc{N}^{\otimes n})}{n + \frac{n - I^{c}(\mc{N}^{\otimes n})}{H^{dis}(\mc{P})}} \leq \frac{\mc{Q}(\mc{N})}{1 + \frac{1 - \mc{Q}(\mc{N})}{H^{dis}(\mc{P})}}$, where the last inequality is due to the super-additivity of coherent information and $\mc{Q}(\mc{N}) = \lim_{n \rightarrow \infty} \frac{I^{c}(\mc{N}^{\otimes n})}{n}$.
\end{proof}

\begin{remark}
    Note that the bound in Theorem~\ref{thm:converse_asymp_dissip} is tighter than the bound in Theorem~\ref{thm:converse_asymp}. In general, for a quantum memory system containing qudits corrupted by a qudit channel $\mc{N}$ satisfying the conditions $C1$ and $C2$ in Lemma~\ref{lemma:strong_dpi} with an ensemble of measurements described in Theorem~\ref{thm:converse_asymp}, the upper bound on storage capacity is $\mathfrak{Q}(\mc{N}) \leq \frac{\mc{Q}(\mc{N}) H^{dis}(\mc{P})}{\log d + H^{dis}(\mc{P}) - \mc{Q}(\mc{N})}$, where $d$ is the qudit dimension.
\end{remark}

For a general projective measurement $\mc{P}$, obtaining closed-form expression of $H^{dis}(\mc{P})$ is not always possible. Therefore, we define noisy orthogonal projectors obeying certain symmetry property allowing us to obtain a closed-form expression of \eqref{eqn:converse_asymp_dissip} in Corollary~\ref{cor:converse_asymp_dissip_zeta}.

\begin{definition}
    \label{def:zeta_noisy_meas}
    The set of noisy orthogonal projectors with parameter $\zeta \in [0, \frac{1}{U}]$ is described by $\mc{P}_\zeta = \{\Pi_u : u \in \{1,\ldots,U\}\}$ satisfying $\sum_{u=1}^{U} \Pi_u = \mathbbm{1}$ such that the condition
    \begin{equation}
        \Tr(\Pi_u \psi) \geq \zeta.
        \label{eqn:zeta_noisy_condition}
    \end{equation}
    holds for all $u \in \{1,\ldots,U\}$ and for any density matrix $\psi$.
\end{definition}

\begin{corollary}
\label{cor:converse_asymp_dissip_zeta}
The storage capacity of a quantum memory system equipped with a set of $\zeta$-noisy orthogonal projectors $\mc{P}_\zeta$ satisfying condition \eqref{eqn:zeta_noisy_condition} is
\begin{equation}
    \mathfrak{Q}(\mc{N}) \leq \frac{\mc{Q}(\mc{N}) H^{dis}_U(\zeta)}{1 + H^{dis}_U(\zeta) - \mc{Q}(\mc{N})},
    \label{eqn:converse_asymp_dissip_zeta}
\end{equation}
where $H^{dis}_U(\zeta) = \log_2U - H(1+\zeta-U \zeta, \zeta, \ldots, \zeta)$ and $H(x_1, \ldots, x_U) = -\sum_{i=1}^{U} x_i \log_2(x_i)$ is the entropy function.
\end{corollary}
\begin{proof}
    For any noisy projective measurement $\mc{P}_\zeta$ with $U$ elements, the completeness relation $\sum_{u = 1}^{U} \Pi_u = \mathbbm{1}$ holds, which implies $\Tr\lt(\sum_{u=1}^U \Pi_u \psi \rt) = \Tr(\psi) = 1$ for any density matrix $\psi$. Therefore,
    \begin{equation}
        \Tr(\Pi_{u} \psi) = 1 - \sum_{w \neq u}^{U} \Tr(\Pi_w \psi) \leq 1 - (U-1) \zeta,
        \label{eqn:zeta_noisy_condition_ub}
    \end{equation}
    for any $u \in \{1, \ldots, U\}$. Therefore, maximum entropy dissipation is $H^{dis}_U(\zeta) = \sup_{\{\mc{P}_\zeta\}, \psi} \lt( H(\psi) - \sum_{u = 1}^U p_{\mathbf{u}}(u) H(\psi_u) \rt)$, where the maximization is over the set of noisy projectors with parameter $\zeta$ and density matrix $\psi$. Due to the symmetry of the constraints \eqref{eqn:zeta_noisy_condition} and \eqref{eqn:zeta_noisy_condition_ub}, it can be observed that maximally mixed state of dimension $U$ with projectors $\Pi_u = (1+\zeta-U\zeta) \ketbra{\varphi_u} + \zeta \sum_{w \neq u}^{U} \ketbra{\varphi_w}$, where $\braket{\varphi_y}{\varphi_z}$ = 0 whenever $x \neq y$, maximize the entropy dissipation, yielding $H^{dis}_U(\zeta) = \log_2U - H(1+\zeta-U \zeta, \zeta, \ldots, \zeta)$.
\end{proof}

\begin{remark}
\label{cor:converse_asymp_dissip_binary}
The storage capacity of a quantum memory described in Corollary~\ref{cor:converse_asymp_dissip_zeta}, {\color{black} but restricted to} two-outcome $\zeta$-noisy projective measurements is bounded as 
\begin{equation}
\mathfrak{Q}(\mc{N}) \leq \frac{\mc{Q}(\mc{N}) (1-h_2(\zeta))}{2 - h_2(\zeta) - \mc{Q}(\mc{N})},
\label{eqn:converse_asymp_dissip_zeta_binary}
\end{equation}
which is obtained by substituting $U=2$ in \eqref{eqn:converse_asymp_dissip_zeta}.
\end{remark}

}
\begin{remark}
    When orthogonal projectors are noiseless, i.e., $\zeta = 0$ the upper bound on the ratio of storage capacity in \eqref{eqn:converse_asymp_dissip_zeta_binary} to quantum capacity $\frac{\mathfrak{Q}(\mc{N})}{Q(\mc{N})} \rightarrow \frac{1}{2}$ as $Q(\mc{N}) \rightarrow 0$, and $\frac{\mathfrak{Q}(\mc{N})}{Q(\mc{N})} \rightarrow 1$ as $Q(\mc{N}) \rightarrow 1$, which indicates that the circuit complexity of a quantum memory using extremely noisy qubits is at least twice as large as that of the noiseless quantum memory.
\end{remark}

\begin{corollary}       
    {\color{black}
    Consider a quantum memory system described in Corollary~\ref{cor:converse_asymp_dissip_zeta} with a set of two-outcome $\zeta$-noisy orthogonal projectors. Suppose each qubit is corrupted independently by depolarizing noise $\mc{N}_{\tilde{p}}$ with parameter $\tilde{p} \in [0, \frac{1}{4})$, then the storage capacity is bounded as     
    \begin{equation}
    \mathfrak{Q}(\mc{N}_{\tilde{p}}) \leq \frac{\mc{Q}_{ub}(\mc{N}_{\tilde{p}}) (1-h_2(\zeta)) }{2-h_2(\zeta) - \mc{Q}_{ub}(\mc{N}_{\tilde{p}})},
    \label{eqn:converse_asymp_depol}
    \end{equation}}
    where 
    \begin{equation}
    \resizebox{0.48\textwidth}{!}{$\mc{Q}_{ub}(\mc{N}_{\tilde{p}}) = \text{conv} \lt\{1-h_2({\tilde{p}}), h_2\lt(\frac{1+\gamma({\tilde{p}})}{2}\rt) - h_2\lt(\frac{\gamma({\tilde{p}})}{2}\rt), 1-4 {\tilde{p}} \rt\}$}, \nonumber
    \end{equation}
    conv$(\cdot)$ is the convex hull, and $h_2(\cdot)$ is the binary entropy function.
\end{corollary}
\begin{proof}
    From the proof of Theorem~\ref{thm:converse_asymp}, the reliable storage problem reduces to the problem of reliable communication. From \cite[eq.~36]{sutter2017approximate}, the quantum capacity of the depolarizing channel is bounded as
    \begin{equation}
     \resizebox{0.44\textwidth}{!}{$\mc{Q}(\mc{N}_{\tilde{p}}) \leq \text{conv}\lt\{ 1 - h({\tilde{p}}), h\lt(\frac{1+\gamma({\tilde{p}})}{2}\rt) - h\lt(\frac{\gamma({\tilde{p}})}{2}\rt), 1-4{\tilde{p}} \rt\}.$} \label{eqn:Q_UB_sutter}
    \end{equation}
    where $\gamma({\tilde{p}}) = 4 (\sqrt{1-\tilde{p}} - 1 + {\tilde{p}})$,
    Substituting \eqref{eqn:Q_UB_sutter} in \eqref{eqn:converse_asymp_dissip_zeta_binary}, we obtain the upper bound in \eqref{eqn:converse_asymp_depol}.
\end{proof}

Thus far, while computing storage capacities we have assumed the number of physical qubits to be arbitrarily large. In the following section, we investigate the limits of reliable information storage in quantum memory systems with a fixed number of physical qubits.

\subsection{Non-asymptotic bounds}

When the number of physical qubits $n$ is fixed, the notions of reliable forever and arbitrarily small logical error rates are not always meaningful. Therefore, instead of stability of memory sequences, we provide an alternative definition of reliability namely $(\epsilon, T)$-reliability of a quantum memory system with a fixed number of physical qubits $n$ as follows.

\begin{definition}
    \label{def:non_asymp}
    A quantum memory $M_{k,n}$ with a storage capability of $k$ qubits constructed using a quantum system containing $n$ physical qubits is said to be $(\epsilon, T)$-reliable if it satisfies the following:
    \begin{enumerate}
        \item At time $t = 0$, each memory $M_{k, n}$ is initialized with a state $\rho_k = \mc{E}(\sigma_k)$, where $\sigma_k \in \mc{B}(\mc{H})$, $\dim(\mc{H}) = 2^k$ and $\mc{E}:\mc{B}(\mc{H}) \rightarrow \mc{B}(\mc{H}')$ is a CPTP map with $\dim(\mathcal{H}') \leq n$, which acts as an encoder.
        \item For the given $\epsilon>0$ and $T>0$, there exists $t \geq T$ and a CPTP map (a decoder) $\mc{D}:\mc{B}(\mc{H}') \rightarrow \mc{B}(\mc{H})$ such that $\underset{\sigma_k}{\min}~F(\mc{D}(\hat{\rho}_k), \sigma_k) \geq 1 - \epsilon$, {\color{black}where $\hat{\rho}_k$ is the state of $M_k$ at time $t$ given that its initial state at time $t=0$ was $\rho_k$}.
    \end{enumerate}
\end{definition}

\begin{definition}
    Quantum storage capacity $\mathfrak{Q}^{(\epsilon, T)}_n$ of an $(\epsilon, T)$-reliable quantum memory system with $n$ physical qubits and circuit complexity $\chi(M_{k, n})$ is a number such that $\mathfrak{Q}^{(\epsilon, T)}_n \geq \frac{k}{\chi(M_{k, n})}$ for all $k$.
\end{definition}

Next, we obtain an upper bound on storage capacity when the qubits experience a general i.i.d. channel, and a tighter upper bound under i.i.d. depolarizing noise. In this subsection, we do not consider measurement noise (i.e., we set $\zeta=0$) since it does not provide significant insights in addition to those discussed in Section~\ref{subsec:asymp_bounds}.

\subsubsection{A converse bound using one-shot quantum capacity} 

From one-shot quantum capacity over a quantum channel \cite[Corollary~14.4]{khatri2020principles}, we obtain the following upper bound on the storage capacity of a quantum memory system affected by a general i.i.d. noise.

\begin{corollary}
    Suppose a quantum memory system containing $n$ physical qubits, with each qubit independently experiencing noise characterized by a qubit channel $\mc{N}$, and a set of two-outcome (noiseless) projective measurements. For a target fidelity of $1-\epsilon$, the storage capacity is bounded above as 
    {\color{black}
    \begin{equation}
        \mathfrak{Q}^{(\epsilon, T)}_n(\mc{N}) \leq \frac{1}{1-2 \epsilon} \lt(\mc{Q}(\mc{N})+\frac{h_2(\epsilon)}{n}\rt), 
        \label{eqn:converse_fbl_oneshot}
    \end{equation}
    }
\end{corollary}
\begin{proof}    
    Using the argument in Theorem~\ref{thm:converse_asymp}, the reliable storage problem reduces to the reliable communication problem. Then, for a fixed number of physical qubits $n$, the storage rate is bounded by the one-shot capacity of a finite blocklength qubit channel as \cite{khatri2020principles}
    \begin{equation}
    (1-2 \epsilon) \log_2(d_l) \leq I^c(\mc{N}^{\otimes n}) + h_2(\epsilon),
\end{equation}
where $d_l$ is the dimension of the logical quantum state, and $I^c(\mc{N}^{\otimes n})$ is the coherent information of the channel $\mc{N}^{\otimes n}$. Since coherent information is super-additive, and noting that quantum capacity is the limit of regularized coherent information yields
\begin{align}
    (1-2 \epsilon) \frac{\log_2(d_l)}{n} &\leq \underset{k}{\sup} \frac{I^c(\mc{N}^{\otimes k})}{k} + \frac{h_2(\epsilon)}{n}, \\
    R & \leq \frac{1}{1 - 2 \epsilon} \lt( \mc{Q}(\mc{N}) + \frac{h_2(\epsilon)}{n} \rt), \label{eqn:converse_one_shot_qcap}
\end{align}
where the rate $R = \frac{\log_2 d_l}{n}$.
\end{proof}

\subsubsection{A tighter converse bound for depolarizing noise}
In the particular case of qubits in a quantum memory system experiencing i.i.d. depolarizing noise $\mathcal{N}_{\tilde{p}}$, we obtain a tighter bound using second-order terms in the finite blocklength quantum communication bound.

\begin{corollary}
    \label{cor:converse_non_asymp_depol}
    Suppose a quantum memory system with $n$ physical qubits experiences i.i.d. depolarizing noise $\mc{N}^{\otimes n}_{\tilde{p}}$ with parameter ${\tilde{p}} \in [0, \frac{1}{4})$, and contains a set of two-outcome projective measurements. For a target fidelity of at least $1-\epsilon$, the storage capacity is bounded above as 
    {\color{black}
    \begin{equation}
    \mathfrak{Q}^{(\epsilon, T)}_n(\mc{N}_{\tilde{p}}) \leq \mc{Q}_{so}(\tilde{p}, n, \epsilon),
    \label{eqn:converse_depol_so}
    \end{equation}
    }
    where $\mc{Q}_{so}(\tilde{p}, n, \epsilon) $ is the convex hull of the following terms:
    \begin{itemize}
        \item $\mc{Q}(\mc{Z}_{\tilde{p}}) + \sqrt{\frac{V(\mc{Z}_{\tilde{p}})}{n}} \Phi^{-1}(\epsilon) + \frac{\log n}{2 n} + O\lt(\frac{1}{n}\rt)$,
        \item $h_2\lt(\frac{1+\gamma({\tilde{p}})}{2}\rt) - h_2\lt(\frac{\gamma({\tilde{p}})}{2}\rt)$, and
        \item $1-4 \tilde{p}$,
    \end{itemize}
    and $\mc{Q}(\mc{Z}_{\tilde{p}})$ and $V(\mc{Z}_{\tilde{p}})$ are the quantum capacity and the channel dispersion (or information variance) of the qubit dephasing channel $\mc{Z}_{\tilde{p}}$ with parameter ${\tilde{p}}$, respectively.
\end{corollary}
\begin{proof}    
    The achievable rate of quantum communication of the depolarizing channel $\mc{N}_{\tilde{p}}$ is \cite[eq.~12]{tomamichel2016quantum}, \cite{tomamichel2013hierarchy} is bounded above by the quantum capacity of the dephasing channel $\mc{Z}_{\tilde{p}}$ with the same parameter ${\tilde{p}}$:
    \begin{align}
    \hspace{-6mm} \resizebox{0.44\textwidth}{!}{$\frac{\log_2(d)}{n}  \leq  \mc{Q}(\mc{Z}_{\tilde{p}}) + \sqrt{\frac{V(\mc{Z}_{\tilde{p}})}{n}} \Phi^{-1}(\epsilon) + \frac{\log n}{2 n} + O\lt(\frac{1}{n}\rt)$}. \label{eqn:depol_ub_dephase}
    \end{align}
    Replacing the first term in the convex hull in \eqref{eqn:Q_UB_sutter} with the above upper bound on the quantum capacity of dephasing channel, we obtain the bound in \eqref{eqn:converse_depol_so}.
\end{proof}

\begin{remark}
    It can be proved that the constant associated with the last term $O(\frac{1}{n})$ in \eqref{eqn:depol_ub_dephase} is negative (see Appendix~\ref{apndx:negative_constant}. for details), Therefore, the sum of first three terms in \eqref{eqn:depol_ub_dephase} is a valid upper bound, which we use in our numerical analyses in Sections~\ref{sec:time_complex_classical} and \ref{sec:compare_lb_ub_non_asymp}.
\end{remark}

In the next section, we investigate how time complexity of classical computation time of the decoder affects the upper bounds on storage capacity.

\section{Considering the time-complexity of classical computation in the refresh block}
\label{sec:time_complex_classical}

Thus far, we have assumed that the time complexity of the classical computation in error detection and correction (in the refresh block) is constant, which is a usual assumption in fault-tolerant quantum computation \cite{gottesman2013fault}. However, in practice, it is unlikely that this assumption holds, i.e., the time complexity scales non-trivially with the number of physical qubits $n$. For example, in \cite{fawzi2020constant}, the classical decoding algorithm (small set-flip algorithm) runs with a time complexity of $O(n)$, and its parallelized version scales as $O(\log n)$.

A side effect of considering the time complexity of classical computation is the inability to achieve an arbitrarily small probability of logical error for a given rate. In other words, the improvement in the logical error provided by large blocklengths is overshadowed by the accumulation of more noise while waiting for the completion of classical computation in each cycle. Consequently, our assumption of $p(\tau)$ being a constant no longer holds, and the time dependence of physical noise parameter $p$ plays an unfavorable, yet important role on the storage capacity. Note that syndrome error $q$ does not depend on time since ancillas are freshly prepared or reset before every syndrome measurement. We consider the model of time dependence of the physical noise parameter $p$ proposed in \cite{etxezarreta2021time}, which is $p(\tau) = \frac{1}{2}-\frac{1}{6}e^{-\tau/\tau_r} - \frac{1}{3}e^{-\tau/\tau_d}$, where $\tau_r$ and $\tau_d$ are the relaxation and dephasing times, respectively. In general, determining storage capacity $\mathfrak{Q}^{(\epsilon, T)}$ under the decoding time constraint can be formulated as the following optimization problem:

\begin{align}
\underset{n, \tau}{\mbox{maximize }} & R^{(\epsilon, T)}(n, \tau) \label{eqn:optimiz_gen} \\
\mbox{s.t. } & c_1 g(n) \tau_c + \tau_0 \leq \tau, \\
& p(\tau) + p_r \leq p_{th}, \\
& P^{(T)}_e(n, p(\tau), p_r, p_{th}) \leq \epsilon.
\label{eqn:optimiz_t}
\end{align}
where $R^{(\epsilon, T)}(n, \tau)$ is an achievable storage rate of a quantum memory containing $n$ physical qubits, $\Omega(g(n))$ is the time complexity of classical computation, $\tau_c$ is the time required for each cycle of classical computation in the decoder, $\tau_0$ is the time required for reading out syndrome measurements, $P^{(T)}_e(n, p(\tau), p_r, p_{th})$ is the logical error probability at time $T$, and $c_1>0$ is a constant. In general, optimization problem \eqref{eqn:optimiz_gen} may not be tractable since precise characterizations of $R^{(\epsilon, T)}(n, \tau)$ and $P^{(T)}_e(n, p(\tau), p_r, p_{th})$ may not be known. However, we can find upper bounds on storage capacities in specific cases; for example, under i.i.d. depolarizing noise, and $g(n) = \log n$, solving the following simpler optimization problem yields an upper bound on storage capacity (see Appendix \ref{apndx:optimiz_depol} for the derivation):
{\color{black}
\begin{align}
\underset{n, \tau}{\mbox{maximize }} & \mc{Q}_{so}(\tilde{p}(\tau), n, \epsilon) \label{eqn:optimiz_depol} \\
\mbox{s.t. } & c_1 g(n) \tau_c + \tau_{0} \leq \tau \leq \tau_{\max},
\end{align}
}
where $\mc{Q}_{so}(\tilde{p}(\tau), n, \epsilon)$ is the finite blocklength upper bound on quantum capacity of the depolarizing channel with parameter $\tilde{p}(\tau) = 3 p(\tau)/2$, and $\tau_{\max} = \max\{\tau : \tilde{p}(\tau) = \frac{1}{4}\}$. As shown in Appendix~\ref{apndx:optimiz_depol}, problem \eqref{eqn:optimiz_depol} can be further reduced to a one-dimensional non-convex optimization problem in $n$ which can be solved optimally using \cite[Algorithm~3]{uthirakalyani2023limits}. Figure~\ref{fig:contour_plot_Qub_depol_n_vs_tau} shows the contours of the objective function, the constraint curve, and the optimal solution $(n^*, \tau^*)$ of optimization problem \eqref{eqn:optimiz_depol} in a geometric manner.

\begin{figure}    
    \includegraphics[scale=0.62]{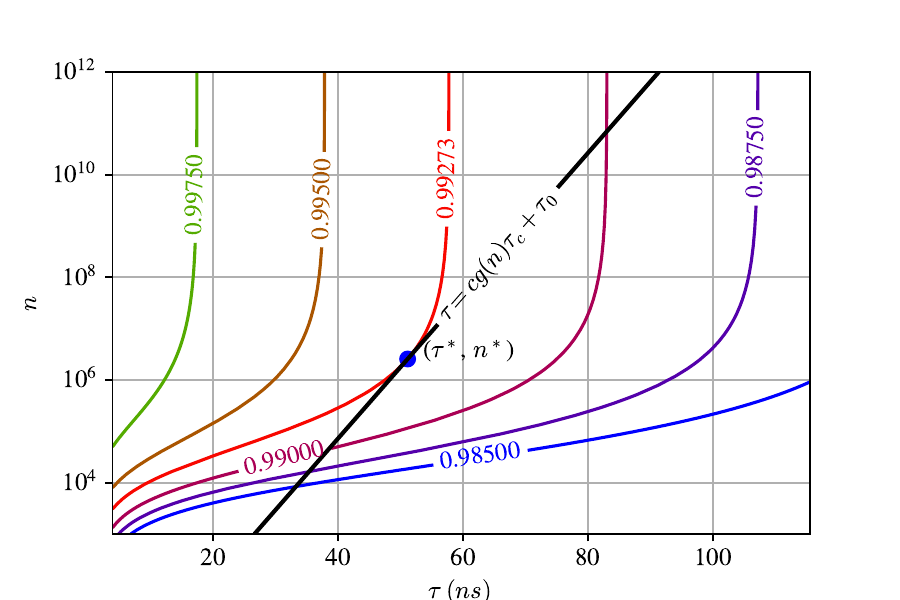}
    \caption{Geometric illustration of the optimization problem in \eqref{eqn:optimiz_depol}. The values of the constants are set as follows: $\tau_r = 49\,\mu$s, $\tau_d = 95 \mu$s \cite{burnett2019decoherence}, $1/\tau_c =$ 3.2 GHz, and $c_1 = 10$. The optimal point $(n^*, \tau^*) \approx (2.565 \times 10^6, 51.12~\text{ns})$ highlighted in the figure corresponds to the upper bound on storage capacity {\color{black} $\mathfrak{Q}^{(\epsilon, T)} \leq 0.9927320704447439$} for a target fidelity $\epsilon = 10^{-6}$.}
    \label{fig:contour_plot_Qub_depol_n_vs_tau}
\end{figure}

\section{Numerical Comparison between Upper and Lower Bounds on Storage Capacity}

\label{sec:compare_lb_ub}

In this section, we compare the upper and lower bounds of storage capacities in both asymptotic and non-asymptotic cases, when the quantum memory system is corrupted by i.i.d. depolarizing noise across the qubits for a single time-step, i.e., $L=1$ (extension to a general $L$ is straightforward). {\color{black}For comparison, we assume that there is no error in the refresh block (syndrome measurement and error correction) implying that the residual error $p_r \approx 0$, and also the measurement noise parameter $\zeta = 0$.}

\subsection{Asymptotic bounds}
\label{sec:compare_lb_ub_asymp}

Recall from \cite[Lemma~29]{fawzi2018efficient} that the lower bound on storage capacity in Theorem~\ref{thm:converse_asymp} (with local stochastic noise parameter $p$) also holds for depolarizing noise with parameter $\tilde{p} = 3p/2$. Consider a quantum expander code described in Section~\ref{subsec:q_expander_codes} with parameters given in Appendix \ref{apndx:qexpander_codes_params}, and let the depolarizing noise parameter be $\tilde{p} = 2.212 \times 10^{-19} \approx 3 p_{th}/2$. From \eqref{eqn:storage_capacity_lb}, the storage capacity is bounded below as $\mathfrak{Q} \geq R_s \approx 0.0004246$. On the other hand from \eqref{eqn:converse_asymp_depol}, the upper bound on storage capacity is $\mathfrak{Q} \leq 1 - 4.17 \times 10^{-17} \approx 1$ for the same physical noise level.

While quantum expander codes provide a non-vanishing lower bound on storage capacity in theory, the gap between achievable storage rate and the converse is too large to be of practical interest. This gap is a consequence of the small fault-tolerance threshold of quantum expander codes obtained using loose percolation theory bounds \cite{fawzi2018efficient}. Therefore, in the next section, we investigate the gap in the non-asymptotic regime, which is of more practical interest.

\subsection{Non-asymptotic bounds}
\label{sec:compare_lb_ub_non_asymp}
Here we consider a quantum memory system with $k=12$ (logical) qubits protected using a bivariate bicycle quantum LDPC code proposed in \cite{bravyi2024high}. In particular, we consider an $[[n,k,d_c]] = [[144,12,12]]$ code, where $n$ is the number of physical qubits and $d_c$ is the distance of the code. The circuit complexity of the resulting quantum memory system is $\chi = 12 n$ (see Appendix~\ref{apndx:qcyc_qldpc_codes_params} for the calculation). Therefore, the lower bound on storage capacity is $\mathfrak{Q}_n^{(\epsilon, T)} \geq \frac{k}{\chi} = 0.006944$, where the logical probability of error $\epsilon = 2.3639 \times 10^{-7}$.

From \eqref{eqn:depol_ub_dephase}, the upper bound on storage rate for a quantum memory system with $n=144$ physical qubits, and the target logical error rate $\epsilon = 2.3639 \times 10^{-7}$ is {\color{black} $\mathfrak{Q}_n^{(\epsilon, T)} \leq 0.8813$}.

\section{Conclusion}
\label{sec:conclusion}

In this paper, we extend the definition of reliable memories to the quantum setting. We prove that the achievable storage rate of a quantum memory system protected using quantum expander codes is a strictly positive constant when both data qubits and error correction are affected by local stochastic noise. Then using an entropy homeostasis argument similar to \cite{varshney2015toward} we provide a tighter upper bound on the storage capacity of a quantum memory system corrupted by noise, which is independent and identically distributed across qubits. We also provide a strictly positive lower bound on the storage capacity using bivariate bicycle LDPC codes proposed in \cite{bravyi2024high} and converse upper bounds in the non-asymptotic regime.

Finally, we compare the achievable storage rates with the upper bounds for specific values of depolarizing noise. In the asymptotic case, even though the lower bound does not vanish asymptotically, there is a large gap between the bounds. On the other hand, in the non-asymptotic regime, we compare the upper bound with the quantum storage rate obtained using a bivariate bicycle quantum LDPC code; the gap between the bounds is relatively smaller. The non-asymptotic bounds are relevant for NISQ technologies and hence are potentially of more practical interest.

As an extension of this work, one can explore other quantum LDPC codes with more efficient decoders to tighten the lower bounds. {\color{black} The noise parameter threshold of quantum expander codes being non-zero is of fundamental significance, but the fact that it is quite small motivates us to search for better constructions}. On the converse side, one can tighten the upper bounds further using a better entropy dissipation argument.  Studying the decay of quantum coherence with time using suitable strong data processing inequalities could also improve upper bounds on storage capacity.

\section{Acknowledgements}
This work was supported in part by National Science Foundation grant PHY-2112890.


\appendix

\section{Approximation term in (\ref{eqn:depol_ub_dephase}) is negative}
\label{apndx:negative_constant}
From \cite[eq.~12]{tomamichel2016quantum}, the finite blocklength quantum capacity of a depolarizing channel with parameter $\tilde{p}$ is bounded above by that of a qubit dephasing channel with the same parameter $\tilde{p}$. Then from \cite[eq.~6]{tomamichel2016quantum}, the bound is equivalent to the (classical) binary symmetric channel (BSC) with parameter $\tilde{p}$. Substituting $r = \frac{1}{4}$ in \cite[eq.~602]{polyanskiy2010channel}, we obtain the constant $\log_2\lt( \frac{1}{1-2r} \rt) = 1$, which has a negative sign in \cite[eq.~587]{polyanskiy2010channel}. \hfill \qedsymbol

\section{Quantum expander codes}




\label{apndx:qexpander_codes_params}
In our numerical analysis, the code parameters are assigned the following values: 
$\gamma = 2$, $\delta = 10^{-5}$, $d_A = 7 $, $d_B = 8$.

Some constants that depend on the above parameters are:
\begin{itemize}    
    \item $d = d_B^2 + 2d_B(d_A-1) $, $r = \frac{d_A}{d_B} $,
    \item $\beta_0 = 1-8\delta $, $\beta_1 = 1-16\delta $, $\beta = 0.99 \beta_1 $, 
    \item $\gamma_0 = \frac{r^2\gamma}{\sqrt{1+r^2}} $, $\alpha = \frac{r\beta}{4+2r\beta}$, $c_0 = \frac{4}{d_A (\beta_1 - \beta)}$, 
    \item $h(\alpha) = -\alpha\log_2(\alpha) - (1-\alpha)\log_2(1-\alpha)$,
    \item $C=\lt[\lt(1-2^{h(\alpha) / \alpha} p\rt)\lt(1-\lt(\frac{p}{p_{th}}\rt)^{\alpha}\rt) \rt]^{-1}$, $C' = \alpha \gamma_0$.
\end{itemize}

The noise threshold corresponding to the parameters is $p_{th} = \lt( \frac{2^{-h(\alpha)}}{(d-1)(1+\frac{1}{d-2})^{(d-2)}} \rt)^{\frac{1}{\alpha}} = 2.212 \times 10^{-19}$.

The residual error after refresh is \cite[proof of Theorem~3]{fawzi2020constant} $p_r = K q^{1/c_0}$, where $K$ is a constant independent of $p$ and $q$.

\section{Bivariate bicycle quantum LDPC codes \cite{bravyi2024high}}
\subsection{Logical error rate}

Please refer to \cite[Section~4]{bravyi2024high} for a detailed explanation of bivariate bicycle quantum LDPC codes. The logical error rate of a $[[144,12,12]]$ bivariate bicycle quantum LDPC code protecting a quantum memory using superconducting qubits with physical noise (depolarizing noise) with parameter $p=10^{-3}$ obtained using a fitting formula (see \cite[Table 6]{bravyi2024high}) is $P_e(n) = p^{d_{circ}/2} e^{18.04 + 1337 p - 96007 p^2} \approx 2.3639 \times 10^{-7}$, where $d_{circ}=10$ and $n = 144$.

\subsection{Computing the circuit complexity of $[[144, 12, 12]]$ bivariate bicycle quantum LDPC code}
\label{apndx:qcyc_qldpc_codes_params}
The circuit complexity of error detection and correction are calculated as follows \cite{bravyi2024high}: number of ancillas $n_a = n$, number of Hadamard gates (after initialization and before measurement) for phase-flip errors, $n_H = 2 (n/2) = n$. Since the degree of every check node is 6, total the number of CNOT gates for bit-flip and phase-flip errors (combined) $n_{CNOT}=2 \cdot 6 \cdot (n/2) = 6n$. The number of Z-basis measurements of the ancillas is $n_m = n_a = n$. Finally, the error correction operation of an individual qubit is either $X$ or $Z$ gate based on the syndrome measurement, which implies that the number of error-correcting gates, $n_{ECC} = 2n$. Therefore, the total number of gates involved in error correction (gate complexity), $\chi_g = n_H + n_{CNOT} + n_{ECC} = n + 6n + 2n = 9n$. The total number of components, $\chi = n + n_a + \chi_g + n_m = 12 n$. The circuit complexity of $[[144,12,12]]$ code $\chi = 1728$.

\section{Storage capacity with decoder time-complexity constraint}
\label{apndx:optimiz_depol}

Since we are interested in an upper bound of the objective function in \eqref{eqn:optimiz_gen}, we assert that the logical probability of error over a single refresh cycle is at most $\epsilon$, which implies that the corresponding fidelity is at least $1-\epsilon$ (Lemma~\ref{lemma:prob_err_fidelity}). For the depolarizing noise we know that quantum capacity vanishes for $\tilde{p}(\tau) > \frac{1}{4}$, therefore, we set the threshold $p_{th} = \frac{1}{6}$. Using the upper bound on storage capacity in Corollary~\ref{cor:converse_non_asymp_depol}, the optimization problem becomes
{\color{black}
\begin{align}
\underset{n, \tau}{\mbox{maximize }} & \mc{Q}_{so}(\tilde{p}(\tau), n, \epsilon) \label{eqn:optimiz_depol_apndx1} \\
\mbox{s.t. } & c_1 g(n) \tau_c + \tau_{0} \leq \tau \leq \tau_{\max},
\end{align}}
where $\mc{Q}_{so}(\tilde{p}(\tau), n, \epsilon)$ is the finite blocklength upper bound on quantum capacity of the depolarizing channel with parameter $\tilde{p}(\tau) = 3 p(\tau)/2$, and $\tau_{\max} = \max\{\tau : \tilde{p}(\tau) = \frac{1}{4}\}$.
It can be observed that $c_1 g(n) \tau_c + \tau_{0} \leq \tau$ is an active constraint. Suppose assume the contrary, i.e., the optimizer is
in the interior of the feasible set, then increasing $n$ or decreasing $\tau$ increases the objective function further, which is a contradiction. Therefore, substituting $\tau(n) = c_1 \log(n) \tau_c + \tau_0$ we rewrite \eqref{eqn:optimiz_depol_apndx1} as
{\color{black}
\begin{equation}
\underset{n \in [1, n_{\max}]}{\mbox{maximize }} \mc{Q}_{so}(\tilde{p}(\tau(n)), n, \epsilon), \label{eqn:optimiz_depol_apndx2}
\end{equation}}
where $n_{\max} = n$ satisfying $\tilde{p}(\tau(n)) = \frac{1}{4}$. Note that \eqref{eqn:optimiz_depol_apndx2} is in general a non-convex optimization problem over a closed interval. If the objective function is continuously differentiable with bounded Lipschitz gradients, then the problem can be solved optimally using a line search algorithm proposed in \cite[Algorithm~3]{uthirakalyani2023limits}.

\section{Refining the upper bound on classical storage capacity}

\label{apndx:classical_storage_capacity_ub}

In \cite[Theorem~5]{varshney2015toward}, the upper bound on classical storage capacity, $\mathfrak{C}(\alpha) \leq \frac{C(\alpha)}{1 + \frac{h_2(\alpha)}{2 -  h_2(\alpha/2 + 1/4)}}$, was obtained using an entropy dissipation argument with $\alpha$-noisy non-trivial extremal logic gates (such as AND, OR, NAND, and NOR) were considered as entropy dissipators, and the input of the gate was assumed to be distributed uniformly over the input alphabet. However, the entropy maximizing distribution may not be uniform. Here, we derive a correct upper bound on classical storage capacity. The maximum entropy dissipated by an extremal logic gate can be obtained by maximizing the difference between the entropies of the input and output as follows:

\begin{equation}
    \Delta h^*(\alpha) = \Delta h(\alpha, p^*) = \underset{p}{\max}~\Delta h(\alpha, p),
    \label{eqn:delta_h_opt}
\end{equation}
where $\Delta h(\alpha, p) = h_2\lt( p, \frac{1-p}{3}, \frac{1-p}{3}, \frac{1-p}{3} \rt) - h_2(\alpha + p - 2 \alpha p)$. Consequently, the corrected upper bound on storage capacity of an $\alpha$-noisy classical memory is
\begin{equation}
    \mathfrak{C}(\alpha) \leq \frac{C(\alpha)}{1 + \frac{h_2(\alpha)}{\Delta h^*(\alpha)}}.
\end{equation}

The closed-form solution to the problem \eqref{eqn:delta_h_opt} is not known. Therefore, we obtain $\Delta h^*(\alpha)$ numerically for every component noise level $\alpha$. The upper bounds on the classical storage capacity  storage are shown in Figure~\ref{fig:classical_storage_capacity_old_vs_new_UB}.

\begin{figure}[H]    
    \includegraphics[scale=0.4]{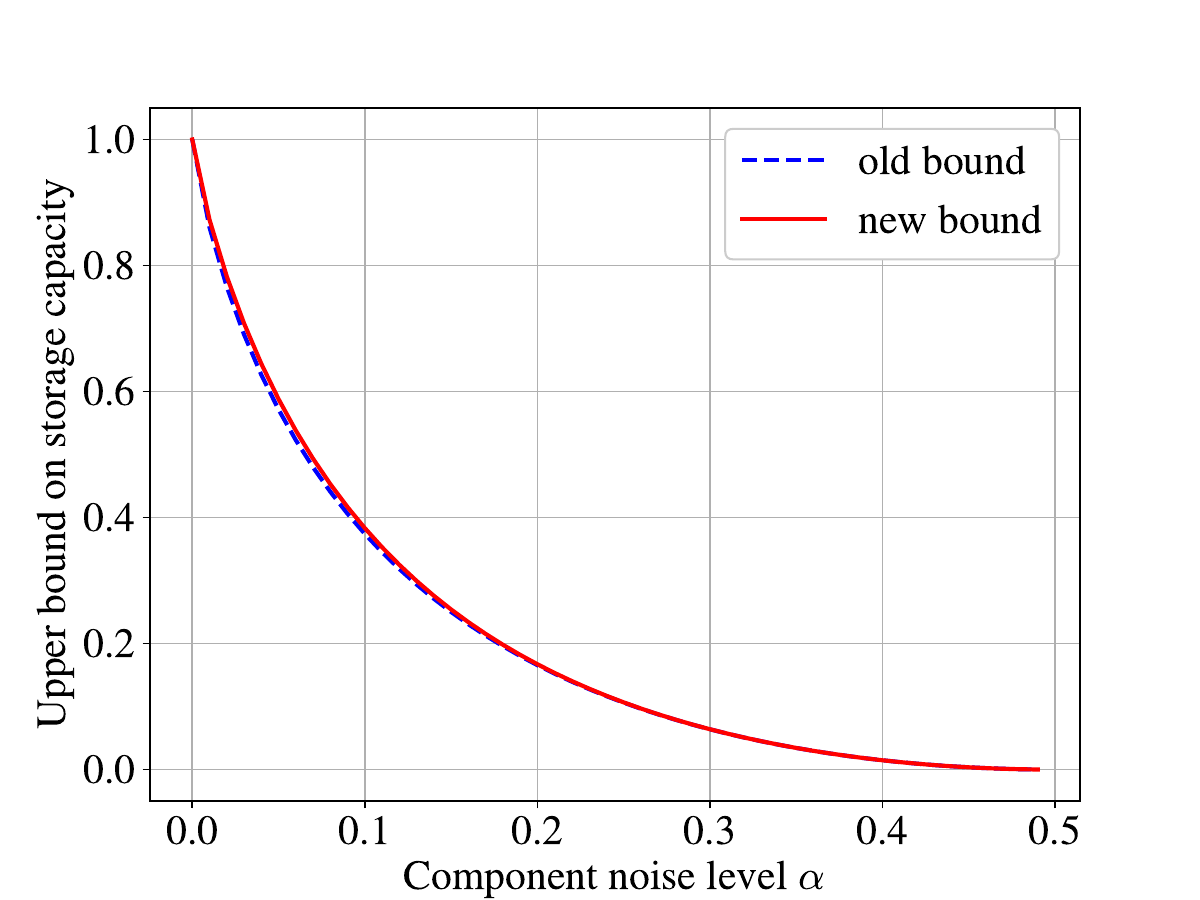}
    \caption{Comparison between the old and new upper bounds on classical storage capacity. The gap between the bounds is relatively significant in the range $\alpha \in [0, 0.2]$.}
    \label{fig:classical_storage_capacity_old_vs_new_UB}
\end{figure}

\nocite{*}
\bibliography{bibfile}

\begin{thebibliography}{38}%
\makeatletter
\providecommand \@ifxundefined [1]{%
 \@ifx{#1\undefined}
}%
\providecommand \@ifnum [1]{%
 \ifnum #1\expandafter \@firstoftwo
 \else \expandafter \@secondoftwo
 \fi
}%
\providecommand \@ifx [1]{%
 \ifx #1\expandafter \@firstoftwo
 \else \expandafter \@secondoftwo
 \fi
}%
\providecommand \natexlab [1]{#1}%
\providecommand \enquote  [1]{``#1''}%
\providecommand \bibnamefont  [1]{#1}%
\providecommand \bibfnamefont [1]{#1}%
\providecommand \citenamefont [1]{#1}%
\providecommand \href@noop [0]{\@secondoftwo}%
\providecommand \href [0]{\begingroup \@sanitize@url \@href}%
\providecommand \@href[1]{\@@startlink{#1}\@@href}%
\providecommand \@@href[1]{\endgroup#1\@@endlink}%
\providecommand \@sanitize@url [0]{\catcode `\\12\catcode `\$12\catcode `\&12\catcode `\#12\catcode `\^12\catcode `\_12\catcode `\%12\relax}%
\providecommand \@@startlink[1]{}%
\providecommand \@@endlink[0]{}%
\providecommand \url  [0]{\begingroup\@sanitize@url \@url }%
\providecommand \@url [1]{\endgroup\@href {#1}{\urlprefix }}%
\providecommand \urlprefix  [0]{URL }%
\providecommand \Eprint [0]{\href }%
\providecommand \doibase [0]{https://doi.org/}%
\providecommand \selectlanguage [0]{\@gobble}%
\providecommand \bibinfo  [0]{\@secondoftwo}%
\providecommand \bibfield  [0]{\@secondoftwo}%
\providecommand \translation [1]{[#1]}%
\providecommand \BibitemOpen [0]{}%
\providecommand \bibitemStop [0]{}%
\providecommand \bibitemNoStop [0]{.\EOS\space}%
\providecommand \EOS [0]{\spacefactor3000\relax}%
\providecommand \BibitemShut  [1]{\csname bibitem#1\endcsname}%
\let\auto@bib@innerbib\@empty
\bibitem [{\citenamefont {Liu}\ \emph {et~al.}(2023)\citenamefont {Liu}, \citenamefont {Wang}, \citenamefont {Stein}, \citenamefont {Ding},\ and\ \citenamefont {Li}}]{qm2023}%
  \BibitemOpen
  \bibfield  {author} {\bibinfo {author} {\bibfnamefont {C.}~\bibnamefont {Liu}}, \bibinfo {author} {\bibfnamefont {M.}~\bibnamefont {Wang}}, \bibinfo {author} {\bibfnamefont {S.~A.}\ \bibnamefont {Stein}}, \bibinfo {author} {\bibfnamefont {Y.}~\bibnamefont {Ding}},\ and\ \bibinfo {author} {\bibfnamefont {A.}~\bibnamefont {Li}},\ }\bibfield  {title} {\bibinfo {title} {Quantum memory: A missing piece in quantum computing units},\ }\href {https://arxiv.org/abs/2309.14432} {\bibfield  {journal} {\bibinfo  {journal} {arXiv:2309.14432}\ } (\bibinfo {year} {2023})}\BibitemShut {NoStop}%
\bibitem [{\citenamefont {Stein}\ \emph {et~al.}(2023)\citenamefont {Stein}, \citenamefont {Sussman}, \citenamefont {Tomesh}, \citenamefont {Guinn}, \citenamefont {Tureci}, \citenamefont {Lin}, \citenamefont {Tang}, \citenamefont {Ang}, \citenamefont {Chakram}, \citenamefont {Li}, \citenamefont {Martonosi}, \citenamefont {Chong}, \citenamefont {Houck}, \citenamefont {Chuang},\ and\ \citenamefont {Demarco}}]{microarch2023}%
  \BibitemOpen
  \bibfield  {author} {\bibinfo {author} {\bibfnamefont {S.}~\bibnamefont {Stein}}, \bibinfo {author} {\bibfnamefont {S.}~\bibnamefont {Sussman}}, \bibinfo {author} {\bibfnamefont {T.}~\bibnamefont {Tomesh}}, \bibinfo {author} {\bibfnamefont {C.}~\bibnamefont {Guinn}}, \bibinfo {author} {\bibfnamefont {E.}~\bibnamefont {Tureci}}, \bibinfo {author} {\bibfnamefont {S.~F.}\ \bibnamefont {Lin}}, \bibinfo {author} {\bibfnamefont {W.}~\bibnamefont {Tang}}, \bibinfo {author} {\bibfnamefont {J.}~\bibnamefont {Ang}}, \bibinfo {author} {\bibfnamefont {S.}~\bibnamefont {Chakram}}, \bibinfo {author} {\bibfnamefont {A.}~\bibnamefont {Li}}, \bibinfo {author} {\bibfnamefont {M.}~\bibnamefont {Martonosi}}, \bibinfo {author} {\bibfnamefont {F.}~\bibnamefont {Chong}}, \bibinfo {author} {\bibfnamefont {A.~A.}\ \bibnamefont {Houck}}, \bibinfo {author} {\bibfnamefont {I.~L.}\ \bibnamefont {Chuang}},\ and\ \bibinfo {author} {\bibfnamefont {M.}~\bibnamefont {Demarco}},\ }\bibfield  {title} {\bibinfo {title} {{H}et{A}rch:
  Heterogeneous microarchitectures for superconducting quantum systems},\ }in\ \href {https://doi.org/10.1145/3613424.3614300} {\emph {\bibinfo {booktitle} {Proceedings of the 56th Annual IEEE/ACM International Symposium on Microarchitecture}}},\ \bibinfo {series and number} {MICRO '23}\ (\bibinfo {year} {2023})\ pp.\ \bibinfo {pages} {539--–554}\BibitemShut {NoStop}%
\bibitem [{\citenamefont {von Neumann}(1956)}]{von1956probabilistic}%
  \BibitemOpen
  \bibfield  {author} {\bibinfo {author} {\bibfnamefont {J.}~\bibnamefont {von Neumann}},\ }\bibfield  {title} {\bibinfo {title} {Probabilistic logics and the synthesis of reliable organisms from unreliable components},\ }in\ \href {http://doi.org/10.1515/9781400882618-003} {\emph {\bibinfo {booktitle} {Automata Studies}}},\ \bibinfo {editor} {edited by\ \bibinfo {editor} {\bibfnamefont {C.~E.}\ \bibnamefont {Shannon}}\ and\ \bibinfo {editor} {\bibfnamefont {J.}~\bibnamefont {McCarthy}}}\ (\bibinfo  {publisher} {Princeton University Press},\ \bibinfo {address} {Princeton, NJ, USA},\ \bibinfo {year} {1956})\ pp.\ \bibinfo {pages} {43--98}\BibitemShut {NoStop}%
\bibitem [{\citenamefont {Taylor}(1968)}]{taylor1968reliable}%
  \BibitemOpen
  \bibfield  {author} {\bibinfo {author} {\bibfnamefont {M.~G.}\ \bibnamefont {Taylor}},\ }\bibfield  {title} {\bibinfo {title} {Reliable information storage in memories designed from unreliable components},\ }\href {http://doi.org/10.1002/j.1538-7305.1968.tb01087.x} {\bibfield  {journal} {\bibinfo  {journal} {Bell System Technical Journal}\ }\textbf {\bibinfo {volume} {47}},\ \bibinfo {pages} {2299} (\bibinfo {year} {1968})}\BibitemShut {NoStop}%
\bibitem [{\citenamefont {Kuznetsov}(1973)}]{kuznietsov1973information}%
  \BibitemOpen
  \bibfield  {author} {\bibinfo {author} {\bibfnamefont {A.~V.}\ \bibnamefont {Kuznetsov}},\ }\bibfield  {title} {\bibinfo {title} {Information storage in a memory assembled from unreliable components},\ }\href@noop {} {\bibfield  {journal} {\bibinfo  {journal} {Problems of Information Transmission}\ }\textbf {\bibinfo {volume} {9}},\ \bibinfo {pages} {100} (\bibinfo {year} {1973})}\BibitemShut {NoStop}%
\bibitem [{\citenamefont {Varshney}(2015)}]{varshney2015toward}%
  \BibitemOpen
  \bibfield  {author} {\bibinfo {author} {\bibfnamefont {L.~R.}\ \bibnamefont {Varshney}},\ }\bibfield  {title} {\bibinfo {title} {Toward limits of constructing reliable memories from unreliable components},\ }in\ \href {https://doi.org/10.1109/itwf.2015.7360745} {\emph {\bibinfo {booktitle} {Proceedings of the 2015 IEEE Information Theory Workshop}}}\ (\bibinfo {year} {2015})\ pp.\ \bibinfo {pages} {114--118}\BibitemShut {NoStop}%
\bibitem [{\citenamefont {Chilappagari}\ and\ \citenamefont {Vasic}(2007)}]{chilappagari2007fault}%
  \BibitemOpen
  \bibfield  {author} {\bibinfo {author} {\bibfnamefont {S.~K.}\ \bibnamefont {Chilappagari}}\ and\ \bibinfo {author} {\bibfnamefont {B.}~\bibnamefont {Vasic}},\ }\bibfield  {title} {\bibinfo {title} {Fault tolerant memories based on expander graphs},\ }in\ \href {https://doi.org/10.1109/itw.2007.4313061} {\emph {\bibinfo {booktitle} {Proceedings of the 2007 IEEE Information Theory Workshop}}}\ (\bibinfo {year} {2007})\ pp.\ \bibinfo {pages} {126--131}\BibitemShut {NoStop}%
\bibitem [{\citenamefont {Varshney}(2011)}]{varshney2011performance}%
  \BibitemOpen
  \bibfield  {author} {\bibinfo {author} {\bibfnamefont {L.~R.}\ \bibnamefont {Varshney}},\ }\bibfield  {title} {\bibinfo {title} {Performance of {LDPC} codes under faulty iterative decoding},\ }\href {https://doi.org/10.1109/tit.2011.2145870} {\bibfield  {journal} {\bibinfo  {journal} {IEEE Trans. Inf. Theory}\ }\textbf {\bibinfo {volume} {57}},\ \bibinfo {pages} {4427} (\bibinfo {year} {2011})}\BibitemShut {NoStop}%
\bibitem [{\citenamefont {Shor}(1995)}]{shor1995scheme}%
  \BibitemOpen
  \bibfield  {author} {\bibinfo {author} {\bibfnamefont {P.~W.}\ \bibnamefont {Shor}},\ }\bibfield  {title} {\bibinfo {title} {Scheme for reducing decoherence in quantum computer memory},\ }\href {https://doi.org/10.1103/physreva.52.r2493} {\bibfield  {journal} {\bibinfo  {journal} {Phys. Rev. A}\ }\textbf {\bibinfo {volume} {52}},\ \bibinfo {pages} {R2493} (\bibinfo {year} {1995})}\BibitemShut {NoStop}%
\bibitem [{\citenamefont {Gingrich}\ \emph {et~al.}(2003)\citenamefont {Gingrich}, \citenamefont {Kok}, \citenamefont {Lee}, \citenamefont {Vatan},\ and\ \citenamefont {Dowling}}]{gingrich2003all}%
  \BibitemOpen
  \bibfield  {author} {\bibinfo {author} {\bibfnamefont {R.~M.}\ \bibnamefont {Gingrich}}, \bibinfo {author} {\bibfnamefont {P.}~\bibnamefont {Kok}}, \bibinfo {author} {\bibfnamefont {H.}~\bibnamefont {Lee}}, \bibinfo {author} {\bibfnamefont {F.}~\bibnamefont {Vatan}},\ and\ \bibinfo {author} {\bibfnamefont {J.~P.}\ \bibnamefont {Dowling}},\ }\bibfield  {title} {\bibinfo {title} {All linear optical quantum memory based on quantum error correction},\ }\href {https://doi.org/10.1103/physrevlett.91.217901} {\bibfield  {journal} {\bibinfo  {journal} {Phys. Rev. Lett.}\ }\textbf {\bibinfo {volume} {91}},\ \bibinfo {pages} {217901} (\bibinfo {year} {2003})}\BibitemShut {NoStop}%
\bibitem [{\citenamefont {Terhal}(2015)}]{terhal2015quantum}%
  \BibitemOpen
  \bibfield  {author} {\bibinfo {author} {\bibfnamefont {B.~M.}\ \bibnamefont {Terhal}},\ }\bibfield  {title} {\bibinfo {title} {Quantum error correction for quantum memories},\ }\href {https://doi.org/10.1103/revmodphys.87.307} {\bibfield  {journal} {\bibinfo  {journal} {Reviews of Modern Physics}\ }\textbf {\bibinfo {volume} {87}},\ \bibinfo {pages} {307} (\bibinfo {year} {2015})}\BibitemShut {NoStop}%
\bibitem [{\citenamefont {Liu}\ \emph {et~al.}(2001)\citenamefont {Liu}, \citenamefont {Dutton}, \citenamefont {Behroozi},\ and\ \citenamefont {Hau}}]{liu2001observation}%
  \BibitemOpen
  \bibfield  {author} {\bibinfo {author} {\bibfnamefont {C.}~\bibnamefont {Liu}}, \bibinfo {author} {\bibfnamefont {Z.}~\bibnamefont {Dutton}}, \bibinfo {author} {\bibfnamefont {C.~H.}\ \bibnamefont {Behroozi}},\ and\ \bibinfo {author} {\bibfnamefont {L.~V.}\ \bibnamefont {Hau}},\ }\bibfield  {title} {\bibinfo {title} {Observation of coherent optical information storage in an atomic medium using halted light pulses},\ }\href {https://doi.org/10.1038/35054017} {\bibfield  {journal} {\bibinfo  {journal} {Nature}\ }\textbf {\bibinfo {volume} {409}},\ \bibinfo {pages} {490} (\bibinfo {year} {2001})}\BibitemShut {NoStop}%
\bibitem [{\citenamefont {Phillips}\ \emph {et~al.}(2001)\citenamefont {Phillips}, \citenamefont {Fleischhauer}, \citenamefont {Mair}, \citenamefont {Walsworth},\ and\ \citenamefont {Lukin}}]{phillips2001storage}%
  \BibitemOpen
  \bibfield  {author} {\bibinfo {author} {\bibfnamefont {D.~F.}\ \bibnamefont {Phillips}}, \bibinfo {author} {\bibfnamefont {A.}~\bibnamefont {Fleischhauer}}, \bibinfo {author} {\bibfnamefont {A.}~\bibnamefont {Mair}}, \bibinfo {author} {\bibfnamefont {R.~L.}\ \bibnamefont {Walsworth}},\ and\ \bibinfo {author} {\bibfnamefont {M.~D.}\ \bibnamefont {Lukin}},\ }\bibfield  {title} {\bibinfo {title} {Storage of light in atomic vapor},\ }\href {https://doi.org/10.1103/physrevlett.86.783} {\bibfield  {journal} {\bibinfo  {journal} {Phys. Rev. Lett.}\ }\textbf {\bibinfo {volume} {86}},\ \bibinfo {pages} {783} (\bibinfo {year} {2001})}\BibitemShut {NoStop}%
\bibitem [{\citenamefont {Schori}\ \emph {et~al.}(2002)\citenamefont {Schori}, \citenamefont {Julsgaard}, \citenamefont {S{\o}rensen},\ and\ \citenamefont {Polzik}}]{schori2002recording}%
  \BibitemOpen
  \bibfield  {author} {\bibinfo {author} {\bibfnamefont {C.}~\bibnamefont {Schori}}, \bibinfo {author} {\bibfnamefont {B.}~\bibnamefont {Julsgaard}}, \bibinfo {author} {\bibfnamefont {J.}~\bibnamefont {S{\o}rensen}},\ and\ \bibinfo {author} {\bibfnamefont {E.}~\bibnamefont {Polzik}},\ }\bibfield  {title} {\bibinfo {title} {Recording quantum properties of light in a long-lived atomic spin state: Towards quantum memory},\ }\href {https://doi.org/10.1103/physrevlett.89.057903} {\bibfield  {journal} {\bibinfo  {journal} {Phys. Rev. Lett.}\ }\textbf {\bibinfo {volume} {89}},\ \bibinfo {pages} {057903} (\bibinfo {year} {2002})}\BibitemShut {NoStop}%
\bibitem [{\citenamefont {Kitaev}(2003)}]{kitaev2003fault}%
  \BibitemOpen
  \bibfield  {author} {\bibinfo {author} {\bibfnamefont {A.~Y.}\ \bibnamefont {Kitaev}},\ }\bibfield  {title} {\bibinfo {title} {Fault-tolerant quantum computation by anyons},\ }\href {https://doi.org/10.1016/s0003-4916(02)00018-0} {\bibfield  {journal} {\bibinfo  {journal} {Annals of Physics}\ }\textbf {\bibinfo {volume} {303}},\ \bibinfo {pages} {2} (\bibinfo {year} {2003})}\BibitemShut {NoStop}%
\bibitem [{\citenamefont {Tillich}\ and\ \citenamefont {Z{\'e}mor}(2013)}]{tillich2013quantum}%
  \BibitemOpen
  \bibfield  {author} {\bibinfo {author} {\bibfnamefont {J.-P.}\ \bibnamefont {Tillich}}\ and\ \bibinfo {author} {\bibfnamefont {G.}~\bibnamefont {Z{\'e}mor}},\ }\bibfield  {title} {\bibinfo {title} {Quantum {LDPC} codes with positive rate and minimum distance proportional to the square root of the blocklength},\ }\href {https://doi.org/10.1109/tit.2013.2292061} {\bibfield  {journal} {\bibinfo  {journal} {IEEE Trans. Inf. Theory}\ }\textbf {\bibinfo {volume} {60}},\ \bibinfo {pages} {1193} (\bibinfo {year} {2013})}\BibitemShut {NoStop}%
\bibitem [{\citenamefont {Gottesman}(2013)}]{gottesman2013fault}%
  \BibitemOpen
  \bibfield  {author} {\bibinfo {author} {\bibfnamefont {D.}~\bibnamefont {Gottesman}},\ }\href {https://arxiv.org/abs/1310.2984} {\bibinfo {title} {Fault-tolerant quantum computation with constant overhead}},\ \bibinfo {howpublished} {arXiv:1310.2984} (\bibinfo {year} {2013})\BibitemShut {NoStop}%
\bibitem [{\citenamefont {Fawzi}\ \emph {et~al.}(2020)\citenamefont {Fawzi}, \citenamefont {Grospellier},\ and\ \citenamefont {Leverrier}}]{fawzi2020constant}%
  \BibitemOpen
  \bibfield  {author} {\bibinfo {author} {\bibfnamefont {O.}~\bibnamefont {Fawzi}}, \bibinfo {author} {\bibfnamefont {A.}~\bibnamefont {Grospellier}},\ and\ \bibinfo {author} {\bibfnamefont {A.}~\bibnamefont {Leverrier}},\ }\bibfield  {title} {\bibinfo {title} {Constant overhead quantum fault tolerance with quantum expander codes},\ }\href {https://doi.org/10.1145/3434163} {\bibfield  {journal} {\bibinfo  {journal} {Communications of the ACM}\ }\textbf {\bibinfo {volume} {64}},\ \bibinfo {pages} {106} (\bibinfo {year} {2020})}\BibitemShut {NoStop}%
\bibitem [{\citenamefont {Fawzi}\ \emph {et~al.}(2022)\citenamefont {Fawzi}, \citenamefont {M{\"u}ller-Hermes},\ and\ \citenamefont {Shayeghi}}]{fawzi2022lower}%
  \BibitemOpen
  \bibfield  {author} {\bibinfo {author} {\bibfnamefont {O.}~\bibnamefont {Fawzi}}, \bibinfo {author} {\bibfnamefont {A.}~\bibnamefont {M{\"u}ller-Hermes}},\ and\ \bibinfo {author} {\bibfnamefont {A.}~\bibnamefont {Shayeghi}},\ }\bibfield  {title} {\bibinfo {title} {A lower bound on the space overhead of fault-tolerant quantum computation},\ }in\ \href {https://doi.org/10.4230/LIPICS.ITCS.2022.68} {\emph {\bibinfo {booktitle} {Proceedings of the 13th Innovations in Theoretical Computer Science Conference (ITCS 2022)}}}\ (\bibinfo {year} {2022})\BibitemShut {NoStop}%
\bibitem [{\citenamefont {Fawzi}\ \emph {et~al.}(2018)\citenamefont {Fawzi}, \citenamefont {Grospellier},\ and\ \citenamefont {Leverrier}}]{fawzi2018efficient}%
  \BibitemOpen
  \bibfield  {author} {\bibinfo {author} {\bibfnamefont {O.}~\bibnamefont {Fawzi}}, \bibinfo {author} {\bibfnamefont {A.}~\bibnamefont {Grospellier}},\ and\ \bibinfo {author} {\bibfnamefont {A.}~\bibnamefont {Leverrier}},\ }\bibfield  {title} {\bibinfo {title} {Efficient decoding of random errors for quantum expander codes},\ }in\ \href {https://doi.org/10.1145/3188745.3188886} {\emph {\bibinfo {booktitle} {Proceedings of the 50th Annual ACM SIGACT Symposium on Theory of Computing}}}\ (\bibinfo {year} {2018})\ pp.\ \bibinfo {pages} {521--534}\BibitemShut {NoStop}%
\bibitem [{\citenamefont {Khatri}\ and\ \citenamefont {Wilde}(2020)}]{khatri2020principles}%
  \BibitemOpen
  \bibfield  {author} {\bibinfo {author} {\bibfnamefont {S.}~\bibnamefont {Khatri}}\ and\ \bibinfo {author} {\bibfnamefont {M.~M.}\ \bibnamefont {Wilde}},\ }\bibfield  {title} {\bibinfo {title} {Principles of quantum communication theory: A modern approach},\ }\href {https://arxiv.org/abs/2011.04672} {\bibfield  {journal} {\bibinfo  {journal} {arXiv:2011.04672v2}\ } (\bibinfo {year} {2020})}\BibitemShut {NoStop}%
\bibitem [{\citenamefont {Gottesman}(2009)}]{gottesman2010introduction}%
  \BibitemOpen
  \bibfield  {author} {\bibinfo {author} {\bibfnamefont {D.}~\bibnamefont {Gottesman}},\ }\bibfield  {title} {\bibinfo {title} {An introduction to quantum error correction and fault-tolerant quantum computation},\ }in\ \href {https://doi.org/10.1090/psapm/068/2762145} {\emph {\bibinfo {booktitle} {Quantum Information Science and its Contributions to Mathematics}}},\ \bibinfo {editor} {edited by\ \bibinfo {editor} {\bibfnamefont {S.~J.}\ \bibnamefont {Lomonaco}, \bibfnamefont {Jr.}}}\ (\bibinfo  {publisher} {American Mathematical Society},\ \bibinfo {year} {2009})\ pp.\ \bibinfo {pages} {13--58}\BibitemShut {NoStop}%
\bibitem [{\citenamefont {Leverrier}\ \emph {et~al.}(2015)\citenamefont {Leverrier}, \citenamefont {Tillich},\ and\ \citenamefont {Z{\'e}mor}}]{leverrier2015quantum}%
  \BibitemOpen
  \bibfield  {author} {\bibinfo {author} {\bibfnamefont {A.}~\bibnamefont {Leverrier}}, \bibinfo {author} {\bibfnamefont {J.-P.}\ \bibnamefont {Tillich}},\ and\ \bibinfo {author} {\bibfnamefont {G.}~\bibnamefont {Z{\'e}mor}},\ }\bibfield  {title} {\bibinfo {title} {Quantum expander codes},\ }in\ \href {https://doi.org/10.1109/focs.2015.55} {\emph {\bibinfo {booktitle} {Proceedings of the 2015 IEEE 56th Annual Symposium on Foundations of Computer Science}}}\ (\bibinfo {year} {2015})\ pp.\ \bibinfo {pages} {810--824}\BibitemShut {NoStop}%
\bibitem [{\citenamefont {Bravyi}\ \emph {et~al.}(2024)\citenamefont {Bravyi}, \citenamefont {Cross}, \citenamefont {Gambetta}, \citenamefont {Maslov}, \citenamefont {Rall},\ and\ \citenamefont {Yoder}}]{bravyi2024high}%
  \BibitemOpen
  \bibfield  {author} {\bibinfo {author} {\bibfnamefont {S.}~\bibnamefont {Bravyi}}, \bibinfo {author} {\bibfnamefont {A.~W.}\ \bibnamefont {Cross}}, \bibinfo {author} {\bibfnamefont {J.~M.}\ \bibnamefont {Gambetta}}, \bibinfo {author} {\bibfnamefont {D.}~\bibnamefont {Maslov}}, \bibinfo {author} {\bibfnamefont {P.}~\bibnamefont {Rall}},\ and\ \bibinfo {author} {\bibfnamefont {T.~J.}\ \bibnamefont {Yoder}},\ }\bibfield  {title} {\bibinfo {title} {High-threshold and low-overhead fault-tolerant quantum memory},\ }\href {https://doi.org/10.1038/s41586-024-07107-7} {\bibfield  {journal} {\bibinfo  {journal} {Nature}\ }\textbf {\bibinfo {volume} {627}},\ \bibinfo {pages} {778} (\bibinfo {year} {2024})}\BibitemShut {NoStop}%
\bibitem [{\citenamefont {Tomamichel}\ \emph {et~al.}(2019)\citenamefont {Tomamichel}, \citenamefont {Bravyi}, \citenamefont {Konig},\ and\ \citenamefont {Gosset}}]{tomamichel2019quantum}%
  \BibitemOpen
  \bibfield  {author} {\bibinfo {author} {\bibfnamefont {M.}~\bibnamefont {Tomamichel}}, \bibinfo {author} {\bibfnamefont {S.}~\bibnamefont {Bravyi}}, \bibinfo {author} {\bibfnamefont {R.}~\bibnamefont {Konig}},\ and\ \bibinfo {author} {\bibfnamefont {D.}~\bibnamefont {Gosset}},\ }\bibfield  {title} {\bibinfo {title} {Quantum advantage with noisy shallow circuits in {3D}},\ }in\ \href {https://doi.org/10.1038/s41567-020-0948-z} {\emph {\bibinfo {booktitle} {Proceedings of the IEEE 60th Annual Symposium on Foundations of Computer Science}}}\ (\bibinfo {year} {2019})\BibitemShut {NoStop}%
\bibitem [{\citenamefont {Uthirakalyani}\ \emph {et~al.}(2023{\natexlab{a}})\citenamefont {Uthirakalyani}, \citenamefont {Nayak},\ and\ \citenamefont {Chatterjee}}]{uthirakalyani2023converse}%
  \BibitemOpen
  \bibfield  {author} {\bibinfo {author} {\bibfnamefont {G.}~\bibnamefont {Uthirakalyani}}, \bibinfo {author} {\bibfnamefont {A.~K.}\ \bibnamefont {Nayak}},\ and\ \bibinfo {author} {\bibfnamefont {A.}~\bibnamefont {Chatterjee}},\ }\bibfield  {title} {\bibinfo {title} {A converse for fault-tolerant quantum computation},\ }\href {https://doi.org/10.22331/q-2023-08-16-1087} {\bibfield  {journal} {\bibinfo  {journal} {Quantum}\ }\textbf {\bibinfo {volume} {7}},\ \bibinfo {pages} {1087} (\bibinfo {year} {2023}{\natexlab{a}})}\BibitemShut {NoStop}%
\bibitem [{\citenamefont {Wilde}(2019)}]{wilde2019classical}%
  \BibitemOpen
  \bibfield  {author} {\bibinfo {author} {\bibfnamefont {M.~M.}\ \bibnamefont {Wilde}},\ }\bibfield  {title} {\bibinfo {title} {From classical to quantum {S}hannon theory},\ }\href {https://arxiv.org/abs/1106.1445} {\bibfield  {journal} {\bibinfo  {journal} {arXiv:1106.1445v8}\ } (\bibinfo {year} {2019})}\BibitemShut {NoStop}%
\bibitem [{\citenamefont {Sutter}\ \emph {et~al.}(2017)\citenamefont {Sutter}, \citenamefont {Scholz}, \citenamefont {Winter},\ and\ \citenamefont {Renner}}]{sutter2017approximate}%
  \BibitemOpen
  \bibfield  {author} {\bibinfo {author} {\bibfnamefont {D.}~\bibnamefont {Sutter}}, \bibinfo {author} {\bibfnamefont {V.~B.}\ \bibnamefont {Scholz}}, \bibinfo {author} {\bibfnamefont {A.}~\bibnamefont {Winter}},\ and\ \bibinfo {author} {\bibfnamefont {R.}~\bibnamefont {Renner}},\ }\bibfield  {title} {\bibinfo {title} {Approximate degradable quantum channels},\ }\href {https://doi.org/10.1109/tit.2017.2754268} {\bibfield  {journal} {\bibinfo  {journal} {IEEE Trans. Inf. Theory}\ }\textbf {\bibinfo {volume} {63}},\ \bibinfo {pages} {7832} (\bibinfo {year} {2017})}\BibitemShut {NoStop}%
\bibitem [{\citenamefont {Tomamichel}\ \emph {et~al.}(2016)\citenamefont {Tomamichel}, \citenamefont {Berta},\ and\ \citenamefont {Renes}}]{tomamichel2016quantum}%
  \BibitemOpen
  \bibfield  {author} {\bibinfo {author} {\bibfnamefont {M.}~\bibnamefont {Tomamichel}}, \bibinfo {author} {\bibfnamefont {M.}~\bibnamefont {Berta}},\ and\ \bibinfo {author} {\bibfnamefont {J.~M.}\ \bibnamefont {Renes}},\ }\bibfield  {title} {\bibinfo {title} {Quantum coding with finite resources},\ }\href {https://doi.org/10.1038/ncomms11419} {\bibfield  {journal} {\bibinfo  {journal} {Nature Communications}\ }\textbf {\bibinfo {volume} {7}},\ \bibinfo {pages} {11419} (\bibinfo {year} {2016})}\BibitemShut {NoStop}%
\bibitem [{\citenamefont {Tomamichel}\ and\ \citenamefont {Hayashi}(2013)}]{tomamichel2013hierarchy}%
  \BibitemOpen
  \bibfield  {author} {\bibinfo {author} {\bibfnamefont {M.}~\bibnamefont {Tomamichel}}\ and\ \bibinfo {author} {\bibfnamefont {M.}~\bibnamefont {Hayashi}},\ }\bibfield  {title} {\bibinfo {title} {A hierarchy of information quantities for finite block length analysis of quantum tasks},\ }\href {https://doi.org/10.1109/tit.2013.2276628} {\bibfield  {journal} {\bibinfo  {journal} {IEEE Trans. Inf. Theory}\ }\textbf {\bibinfo {volume} {59}},\ \bibinfo {pages} {7693} (\bibinfo {year} {2013})}\BibitemShut {NoStop}%
\bibitem [{\citenamefont {Etxezarreta~Martinez}\ \emph {et~al.}(2021)\citenamefont {Etxezarreta~Martinez}, \citenamefont {Fuentes}, \citenamefont {Crespo},\ and\ \citenamefont {Garcia-Frias}}]{etxezarreta2021time}%
  \BibitemOpen
  \bibfield  {author} {\bibinfo {author} {\bibfnamefont {J.}~\bibnamefont {Etxezarreta~Martinez}}, \bibinfo {author} {\bibfnamefont {P.}~\bibnamefont {Fuentes}}, \bibinfo {author} {\bibfnamefont {P.}~\bibnamefont {Crespo}},\ and\ \bibinfo {author} {\bibfnamefont {J.}~\bibnamefont {Garcia-Frias}},\ }\bibfield  {title} {\bibinfo {title} {Time-varying quantum channel models for superconducting qubits},\ }\href {https://doi.org/10.1038/s41534-021-00448-5} {\bibfield  {journal} {\bibinfo  {journal} {npj Quantum Information}\ }\textbf {\bibinfo {volume} {7}},\ \bibinfo {pages} {115} (\bibinfo {year} {2021})}\BibitemShut {NoStop}%
\bibitem [{\citenamefont {Uthirakalyani}\ \emph {et~al.}(2023{\natexlab{b}})\citenamefont {Uthirakalyani}, \citenamefont {Nayak}, \citenamefont {Chatterjee},\ and\ \citenamefont {Varshney}}]{uthirakalyani2023limits}%
  \BibitemOpen
  \bibfield  {author} {\bibinfo {author} {\bibfnamefont {G.}~\bibnamefont {Uthirakalyani}}, \bibinfo {author} {\bibfnamefont {A.~K.}\ \bibnamefont {Nayak}}, \bibinfo {author} {\bibfnamefont {A.}~\bibnamefont {Chatterjee}},\ and\ \bibinfo {author} {\bibfnamefont {L.~R.}\ \bibnamefont {Varshney}},\ }\bibfield  {title} {\bibinfo {title} {Limits of fault tolerance on resource-constrained quantum circuits for classical problems},\ }\href {https://doi.org/10.1103/PhysRevA.108.052425} {\bibfield  {journal} {\bibinfo  {journal} {Phys. Rev. A}\ }\textbf {\bibinfo {volume} {108}},\ \bibinfo {pages} {052425} (\bibinfo {year} {2023}{\natexlab{b}})}\BibitemShut {NoStop}%
\bibitem [{\citenamefont {Burnett}\ \emph {et~al.}(2019)\citenamefont {Burnett}, \citenamefont {Bengtsson}, \citenamefont {Scigliuzzo}, \citenamefont {Niepce}, \citenamefont {Kudra}, \citenamefont {Delsing},\ and\ \citenamefont {Bylander}}]{burnett2019decoherence}%
  \BibitemOpen
  \bibfield  {author} {\bibinfo {author} {\bibfnamefont {J.~J.}\ \bibnamefont {Burnett}}, \bibinfo {author} {\bibfnamefont {A.}~\bibnamefont {Bengtsson}}, \bibinfo {author} {\bibfnamefont {M.}~\bibnamefont {Scigliuzzo}}, \bibinfo {author} {\bibfnamefont {D.}~\bibnamefont {Niepce}}, \bibinfo {author} {\bibfnamefont {M.}~\bibnamefont {Kudra}}, \bibinfo {author} {\bibfnamefont {P.}~\bibnamefont {Delsing}},\ and\ \bibinfo {author} {\bibfnamefont {J.}~\bibnamefont {Bylander}},\ }\bibfield  {title} {\bibinfo {title} {Decoherence benchmarking of superconducting qubits},\ }\href {https://doi.org/10.1038/s41534-019-0168-5} {\bibfield  {journal} {\bibinfo  {journal} {npj Quantum Information}\ }\textbf {\bibinfo {volume} {5}},\ \bibinfo {pages} {54} (\bibinfo {year} {2019})}\BibitemShut {NoStop}%
\bibitem [{\citenamefont {Polyanskiy}\ \emph {et~al.}(2010)\citenamefont {Polyanskiy}, \citenamefont {Poor},\ and\ \citenamefont {Verd{\'u}}}]{polyanskiy2010channel}%
  \BibitemOpen
  \bibfield  {author} {\bibinfo {author} {\bibfnamefont {Y.}~\bibnamefont {Polyanskiy}}, \bibinfo {author} {\bibfnamefont {H.~V.}\ \bibnamefont {Poor}},\ and\ \bibinfo {author} {\bibfnamefont {S.}~\bibnamefont {Verd{\'u}}},\ }\bibfield  {title} {\bibinfo {title} {Channel coding rate in the finite blocklength regime},\ }\href {https://doi.org/10.1109/tit.2010.2043769} {\bibfield  {journal} {\bibinfo  {journal} {IEEE Trans. Inf. Theory}\ }\textbf {\bibinfo {volume} {56}},\ \bibinfo {pages} {2307} (\bibinfo {year} {2010})}\BibitemShut {NoStop}%
\bibitem [{\citenamefont {Richardson}\ and\ \citenamefont {Urbanke}(2008)}]{richardson2008modern}%
  \BibitemOpen
  \bibfield  {author} {\bibinfo {author} {\bibfnamefont {T.}~\bibnamefont {Richardson}}\ and\ \bibinfo {author} {\bibfnamefont {R.}~\bibnamefont {Urbanke}},\ }\href {https://doi.org/10.1017/cbo9780511791338} {\emph {\bibinfo {title} {Modern Coding Theory}}}\ (\bibinfo  {publisher} {Cambridge University Press},\ \bibinfo {year} {2008})\BibitemShut {NoStop}%
\bibitem [{\citenamefont {Dennis}\ \emph {et~al.}(2002)\citenamefont {Dennis}, \citenamefont {Kitaev}, \citenamefont {Landahl},\ and\ \citenamefont {Preskill}}]{dennis2002topological}%
  \BibitemOpen
  \bibfield  {author} {\bibinfo {author} {\bibfnamefont {E.}~\bibnamefont {Dennis}}, \bibinfo {author} {\bibfnamefont {A.}~\bibnamefont {Kitaev}}, \bibinfo {author} {\bibfnamefont {A.}~\bibnamefont {Landahl}},\ and\ \bibinfo {author} {\bibfnamefont {J.}~\bibnamefont {Preskill}},\ }\bibfield  {title} {\bibinfo {title} {Topological quantum memory},\ }\href {https://doi.org/10.1063/1.1499754} {\bibfield  {journal} {\bibinfo  {journal} {Journal of Mathematical Physics}\ }\textbf {\bibinfo {volume} {43}},\ \bibinfo {pages} {4452} (\bibinfo {year} {2002})}\BibitemShut {NoStop}%
\bibitem [{\citenamefont {Leditzky}\ \emph {et~al.}(2017)\citenamefont {Leditzky}, \citenamefont {Datta},\ and\ \citenamefont {Smith}}]{leditzky2017useful}%
  \BibitemOpen
  \bibfield  {author} {\bibinfo {author} {\bibfnamefont {F.}~\bibnamefont {Leditzky}}, \bibinfo {author} {\bibfnamefont {N.}~\bibnamefont {Datta}},\ and\ \bibinfo {author} {\bibfnamefont {G.}~\bibnamefont {Smith}},\ }\bibfield  {title} {\bibinfo {title} {Useful states and entanglement distillation},\ }\href {https://doi.org/10.1109/tit.2017.2776907} {\bibfield  {journal} {\bibinfo  {journal} {IEEE Trans. Inf. Theory}\ }\textbf {\bibinfo {volume} {64}},\ \bibinfo {pages} {4689} (\bibinfo {year} {2017})}\BibitemShut {NoStop}%
\bibitem [{\citenamefont {Smith}\ and\ \citenamefont {Smolin}(2008)}]{smith2008additive}%
  \BibitemOpen
  \bibfield  {author} {\bibinfo {author} {\bibfnamefont {G.}~\bibnamefont {Smith}}\ and\ \bibinfo {author} {\bibfnamefont {J.~A.}\ \bibnamefont {Smolin}},\ }\bibfield  {title} {\bibinfo {title} {Additive extensions of a quantum channel},\ }in\ \href {https://doi.org/10.48660/08060205} {\emph {\bibinfo {booktitle} {2008 IEEE Information Theory Workshop}}}\ (\bibinfo {organization} {IEEE},\ \bibinfo {year} {2008})\ pp.\ \bibinfo {pages} {368--372}\BibitemShut {NoStop}%
\end{thebibliography}%

\end{document}